\documentclass[a4]{llncs}
\PassOptionsToPackage{normalem}{ulem}
\usepackage{ulem}

\usepackage{algorithm}
\usepackage{algpseudocode}
\usepackage{graphicx}
\usepackage{epsfig}
\usepackage{caption}
\usepackage{subcaption}
\usepackage{amsmath}
\usepackage{xparse}

\usepackage[unicode=true,pdfusetitle,
 bookmarks=true,bookmarksnumbered=false,bookmarksopen=false,
 breaklinks=false,pdfborder={0 0 1},backref=false,colorlinks=false]
 {hyperref}
\hypersetup{linkcolor=blue}

\title{A Polynomial Time Algorithm for Minimax-Regret Evacuation on a Dynamic Path}
\author{Guru Prakash Arumugam\inst{1} John Augustine\inst{1}\\ Mordecai Golin\inst{2}  Prashanth Srikanthan\inst{1}}
\institute{Indian Institute of Technology Madras
\and Hong Kong University of Science and Technology}

\begin{document}
\maketitle

\begin{abstract}
  A {\em dynamic path network} is an undirected path with evacuees situated at each vertex. To evacuate the path, evacuees travel towards a designated sink (doorway) to exit.  Each edge has a {\em capacity}, the number of evacuees that can enter the edge in unit time.  Congestion occurs if an evacuee has to wait at a vertex  for other evacuees to leave first.

The basic problem is to 
place $k$ sinks on the line, with an associated evacuation strategy, so as to minimize the total time needed to evacuate everyone.

The {\em minmax-regret} version  introduces uncertainty into the input, with the number of evacuees at vertices only being specified to within a range. The problem is to find a universal solution whose regret (difference from optimal for a given input) is minimized over all legal inputs.

The previously best known algorithms for the minmax-regret version problem ran in time exponential in $k$. In this paper, we derive new properties of solutions that yield  the first polynomial time algorithms for solving the problem.

\end{abstract}

\section{Introduction}
\label{Sec:Intro}


{\em Dynamic flow networks} model movement of items on a graph.  The process starts with  each vertex $v$ assigned  some initial set of supplies $w_v$.  Supplies flow across edges from one vertex to another. Each edge $e$ has  a  given capacity $c_e$ which  limits the rate of the flow of supplies into the edge in each time unit.  If all edges have the same capacity $c_e=c$ the network has {\em uniform capacity}. Each $e$ also has a time required to travel the edge.  Note that as supplies move around the graph,
 {\em congestion} can occur, as supplies back up at a vertex.

Dynamic flow networks  were introduced 
in \cite{Ford1958a}   and have since been extensively analyzed.  
One well studied problem on such networks is the transshipment  problem, e.g., \cite{Hoppe2000b}, in which  the graph has several sources and sinks, with the original supplies being the sources and each sink having a specified demand.  The problem is then to find the minimum time required to satisfy all of the demands.

Dynamic Flow problems also model \cite{Higashikawa2014} {\em evacuation problems}.  In these, the vertex supplies are people in a building(s) and the problem is to find a routing strategy (evacuation plan) that evacuates all of them to specified sinks (exits)  in minimum time. 
Note that in these problems the evacuation plan is vertex based.  That is, each vertex has one associated evacuation edge; all people starting or arriving at that vertex must evacuate through that edge. After traversing the edge they arrive at another vertex and traverse its evacuation edge.  This continues until a sink is reached and the people exit. The basic optimization problem is to determine  a plan that minimizes the total time needed to evacuate all the people.

 In some versions of the problem the sinks are known in advance.  In others, such as the ones we will address in this paper, the placement of the sink(s) is part of the evacuation plan with only $k$,  the number of allowed sinks, being specified as part of the problem. To the best of our knowledge there is no known polynomial time algorithm for solving this problem on a general graph. \cite{Mamada2006} gives an $O(n \log^2 n)$ algorithm for solving the
 $1$-sink problem on a dynamic {\em tree} network with general capacities. \cite{Higashikawaa} improves this down to to $O(n \log n)$ when the capacities are uniform. \cite{Higashikawa2014} shows how to solve the $k$-sink problem on a uniform capacity dynamic {\em path network} in $O(k n \log n)$ time.

In practice, the {\em exact} input, e.g.,  number of people $w_v$ at  each vertex $v$, is unknown at the time the plan is drawn up. All that may be known is that $w_v \in [w_v^-, w_v^+]$  for some specified {\em range}.

 One model for attacking this type of uncertainty is to define the {\em regret} of a plan on a particular fixed input as the discrepancy between the evacuation time for that plan on that  input and the minimum time needed to evacuate for that input. The maximum regret of the plan is then taken over all possible inputs.  The {\em minmax-regret} plan is the one that minimizes the maximum regret.
 
 Minmax regret optimization has been extensively studied for the $k$-median (\cite{conf/cocoon/BhattacharyaK12} is a recent case) and many other optimization problems (\cite{Kouvelis1997} provides an introduction to the literature).  The $1$-sink minmax-regret evacuation problem on a uniform capacity path was originally solved in 
  $O(n \log^2 n)$ time by  \cite{ChengHKNSX13}. 
This was reduced down to $O(n \log n)$ by \cite{Wang2013}.  \cite{Higashikawa2014} provides an $O(n \log^2 n)$ algorithm for the $1$-sink minmax-regret problem on a uniform capacity tree.

Returning to the minimax-regret problem uniform capacity {\em path} case, the only algorithms known for $k > 1$ were \cite{Li2014} which for $k=2$ gave  an $O(n^3 \log n)$ algorithm  and \cite{Ni2014a} which gave 
an $O\left(n^{1+k} (\log n)^{1 + \log k}\right)$ algorithm for general $k$.

In this paper we derive many new properties of the the minmax-regret  uniform capacity path problem.
These lead to two new algorithms, both of which improve on the previous ones. The first, which is better for small $k$, runs in  $O(k n^2 (\log n)^k)$. The second, better for larger $k$, runs in   $O(k n^3 \log n)$, which is the first polynomial time algorithm for this problem.

 The paper is structured as follows.  
 Section \ref{Sec:Defs} introduces the setting and definitions.
 Section \ref{Sec:Basic} discusses how to solve the optimal non-regret problem (whose solution will be needed later).
 Section \ref{Sec:MinMax} derives many of the minmax regret properties needed.
 Section \ref{Sec:Alg1} gives our first algorithm, good for small $k$.
 Section \ref{Sec:Rec} derives more properties.
 Section \ref{Sec:Alg2} gives our second, polynomial time, algorithm.

\section{Preliminary Definitions}
\label{Sec:Defs}

\newcommand*{\Scale}[2][4]{\scalebox{#1}{$#2$}}%

\subsection{Model Definition}
Consider a path $P=(V,E)$ with $(n+1)$ vertices(buildings)
$V=\{x_{0},x_{1},...,x_{n}\}$ and $n$ edges(roads) $E=\{e_{1},e_{2},...,e_{n}\}$ with each
edge $e_{i}=(x_{i-1},x_{i})$. $x_{i}$ also denotes the line-coordinate of the
$i^{th}$ building and $x_0 < x_1 < ... < x_n$.

The distance between two vertices $x_{i}$ and $x_{j}$ is given by
$|x_{i}-x_{j}|$. Each vertex $x_{i}\in V$ is associated with
an open interval of weights $[w_{i}^{-},w_{i}^{+}]$, denoting the range within which the number of evacuees in building $x_i$ can lie ($0 < w_i^- \leq w_i^+$). Each edge has a capacity $c$. We define $\tau$ as the time taken to travel a unit distance on the edge.

Let $\mathcal{S}$ denote the Cartesian product of all weight intervals for
$0\leq i\leq n$:

\[ \mathcal{S}=\prod_{0\leq i\leq n}[w_{i}^{-},w_{i}^{+}].
\]  

A {\em scenario} $s\in\mathcal{S}$ is an assignment of weights (number of
evacuees) to all vertices. The weight of a vertex $x_{i}$ under scenario $s$ is denoted by $w_{i}(s)$.

In the paper, we refer to a function $f(x)$ as ``unimodal with a unique minimum value'' if there exists an $m$ such the function $f(x)$ is monotonically decreasing for $x\leq m$ and monotonically increasing for $x\geq m$. The minimum value of the function is attained for $x=m$.

\subsection{Evacuation Time}
Consider a path $P=(V,E)$, with uniform edge capacity $c$, and time to travel unit distance $\tau$. We are given a scenario $s\in\mathcal{S}$, i.e., vertex $x_i$ has weight $w_i(s)$.

\subsubsection{Left and Right Evacuation}
Consider a single sink $x\in P$. 

If there were no  other people in the way, the evacuees at $v_i$ could complete evacuating to $x$ in
time \cite{Kamiyama:2006:EAE:2100322.2100346} $|x_i-x| \tau + \lceil w_i(s)/c \rceil -1$.

In reality, people to the left and right of the sink $x$ need to evacuate to $x$ (as shown in Fig. \ref{fig:1-sink-evac}) and the actual time required must take congestion into account.

We define $\Theta_L(P,x,s)$(resp. $\Theta_R(P,x,s)$) to be the time taken by people to the left(resp. right) of sink $x$ to evacuate to $x$ under scenario $s$. The left(resp. right) evacuation time can be expressed as the maximum of an \emph{evacuation function} of the nodes to the left(resp. right) of the sink. The exact expression, taken from \cite{Kamiyama:2006:EAE:2100322.2100346} are as follows:

\begin{equation} \Theta_{L}(P,x,s)=\max_{x_i < x}\left\{
(x-x_{i})\tau+\left\lceil \frac{\sum_{0\leq j\leq i}w_{j}(s)}{c}\right\rceil
-1\right\} \label{eq:left-evac-c}
\end{equation}

\begin{equation} \Theta_{R}(P,x,s)=\max_{x_i > x}\left\{
(x_{i}-x)\tau+\left\lceil \frac{\sum_{i\leq j\leq n}w_{j}(s)}{c}\right\rceil
-1 \right\} \label{eq:right-evac-c}
\end{equation}
where the maximum is taken over the \emph{evacuation function} of every node.

For exposition, we simplify the equations (as done in \cite{ChengHKNSX13,Ni2014a}) by assuming $c=1$ and omitting the constant (i.e., -1). Setting  $c=1$ will help us simplify some of the proofs later in the paper (but note that the case of $c>1$ can be treated essentially in the same manner). Thus the left and right evacuation times, $\Theta_L(P,x,s)$ and $\Theta_R(P,x,s)$, are redefined as follows:

\begin{equation}\Theta_{L}(P,x,s)=\max_{x_i < x}\left\{
(x-x_{i})\tau+\sum_{0\leq j\leq i}w_{j}(s) \right\} \label{eq:left-evac}
\end{equation}
\begin{equation}
\Theta_{R}(P,x,s)=\max_{x_i > x}\left\{
(x_{i}-x)\tau+\sum_{i\leq j\leq n}w_{j}(s) \right\} \label{eq:right-evac}
\end{equation}

\subsubsection{1-Sink Evacuation}
The evacuation time to sink $x$ is the maximum of the left and right evacuation times, $\Theta^1(P,x,s)=\max\left\{ \Theta_{L}(P,x,s),\Theta_{R}(P,x,s)\right\}$. The superscript ``1'' denotes that this is a 1-Sink problem. 

\begin{figure}[tp]
\centering
\begin{subfigure}[b]{0.45\textwidth}
\includegraphics[scale=0.4]{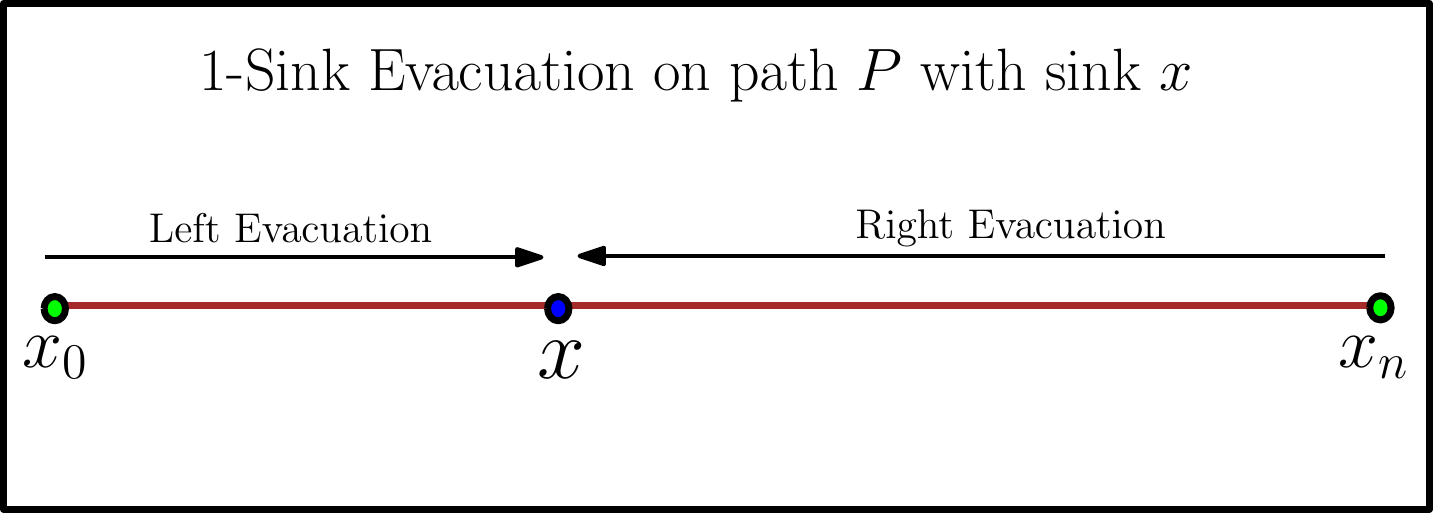}
\caption{$1$-sink evacuation}
\label{fig:1-sink-evac}
\end{subfigure}
\hfill
\begin{subfigure}[b]{0.45\textwidth}
\includegraphics[scale=0.4]{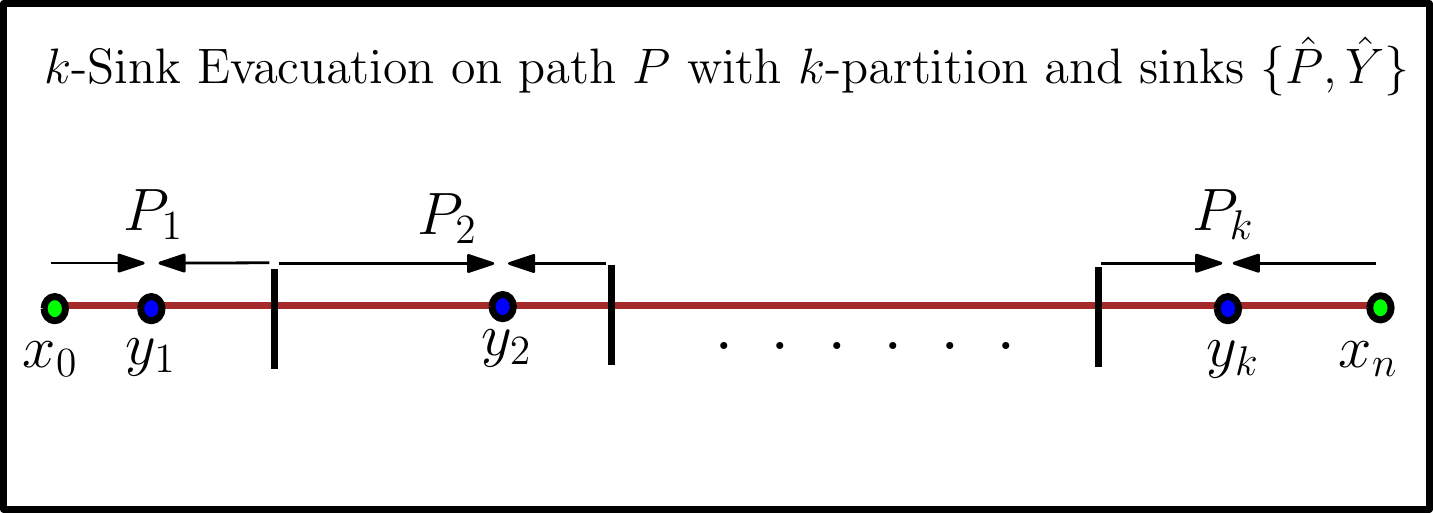}
\caption{$k$-sink evacuation}
\label{fig:k-sink-evac}
\end{subfigure}
\caption{$1$-sink and $k$-sink evacuation illustrations}
\end{figure}

\subsubsection{$k$-Sink Evacuation \label{sec:k-sink-evac}}
Naturally extending the 1-Sink evacuation, suppose we are given a $k$-partition of the path $P$ into $k$ subpaths (or parts) $\hat{P}=\{P_1,P_2,...,P_k\}$ and a set of sinks $\hat{Y}=\{y_1,y_2,...,y_k\}$ such that each $y_i\in P_i$. The evacuees in part $P_i$ are restricted to evacuate only to sink $y_i$ (See Fig. \ref{fig:k-sink-evac}). Every node should completely belong to one part, i.e., no $k$-partition is allowed to split a node. This restricts the people in a building to evacuate to one sink in their part without confusion. The $k$-sink evacuation time will essentially be the maximum of the evacuation time among the individual parts.

Let $\Theta^k(P,\{\hat{P},\hat{Y}\},s)$ denote the $k$-Sink evacuation time on path $P$ with $k$-partition and sinks $\{\hat{P},\hat{Y}\}$ under scenario $s$. Since the evacuation in each subpath $P_i$ are independent of each other, $\Theta^k(P,\{\hat{P},\hat{Y}\},s)$ can be expressed as:

\begin{equation} \Theta^k(P,\{\hat{P},\hat{Y}\},s)=\max_{1\leq i\leq k}
\Theta^1(P_{i},y_i,s) \label{eq:k-sink-evac}
\end{equation}

\begin{definition} Let $d$ be the smallest (and therefore unique) index among the indices of the part(s) which maximizes the term $\Theta^1(P_{i},y_i,s)$. The evacuation time of subpath $P_d$ determines the $k$-Sink evacuation time under scenario $s$. We refer to $P_d$ as the ``dominant part'' under scenario $s$.
\end{definition}

\subsection{Optimal (Non-Regret) $k$-Sink Location Problem}
The Optimal $k$-Sink Location Problem can be defined as follows. The input is a path $P$ with edge capacities $c$, time to travel unit distance $\tau$, and a scenario $s\in\mathcal{S}$.

Let $\Theta^k_{\mathrm{opt}}(P,s)$ be the optimal (minimum) $k$-Sink evacuation time in path $P$ under scenario $s\in\mathcal{S}$. The algorithm needs to provide the $k$-partition and sinks $\{\hat{P},\hat{Y}\}$ as defined in Sect. \ref{sec:k-sink-evac}, which achieves this minimum evacuation time.

\subsection{Regret}
For a choice of $k$-partition and sinks $\{\hat{P},\hat{Y}\}$ and a scenario $s\in\mathcal{S}$, the regret is defined as the difference between the $k$-Sink evacuation time for this choice of $\{\hat{P},\hat{Y}\}$ and the optimal $k$-Sink evacuation time. The regret can be expressed as:

\begin{equation} R(\{\hat{P},\hat{Y}\},s) = \Theta^k(P,\{\hat{P},\hat{Y}\},s)-\Theta^k_{\mathrm{opt}}(P,s) \label{eq:regret}
\end{equation}

The maximum regret (called \emph{max-regret}) achieved (over all scenarios) for a choice of $\{\hat{P},\hat{Y}\}$ is:
\begin{equation} R_{\max}(\{\hat{P},\hat{Y}\}) = \max_{s\in\mathcal{S}}\left\{ R(\{\hat{P},\hat{Y}\},s)  \right\} \label{eq:max-regret}
\end{equation}

If $R_{\max}(\{\hat{P},\hat{Y}\})=R(\{\hat{P},\hat{Y}\},s^*)$ for some scenario $s^*\in\mathcal{S}$, then $s^*$ is called a \emph{worst-case scenario} for $\{\hat{P},\hat{Y}\}$. The minimax-regret for the problem is the minimum \emph{max-regret} over all possible $k$-partitions and sinks $\{\hat{P},\hat{Y}\}$.

\subsection{Minimax-Regret $k$-Sink Location Problem}
The input for the Minimax-Regret $k$-Sink Location Problem is a dynamic path network with path $P$, vertex weight intervals $[w_i^-,w_i^+]$, edge capacity $c$, and time to travel unit distance $\tau$. The problem can be understood as a 2-person game between the algorithm $A$ and the adversary $B$ as follows:

\begin{enumerate}
  \item The algorithm $A$ provides a $k$-partition and sinks $\{\hat{P},\hat{Y}\}$ as defined in Sect. \ref{sec:k-sink-evac}.
  \item The adversary $B$ now chooses a regret-maximizing \emph{worst-case scenario} $s^*\in\mathcal{S}$ with a \emph{max-regret} of $R_{\max}(\{\hat{P},\hat{Y}\})$.
  \item The objective of $A$ is to find the $k$-partition and sinks $\{\hat{P},\hat{Y}\}$ that minimizes the \emph{max-regret}.
\end{enumerate}

\subsection{Sink on Vertex Assumption}
All our definitions until now are generic in the fact that a sink can lie either at a vertex or on an edge (i.e., anywhere on the path). From this point onwards, we are going to assume that a sink always lies at a vertex (and minmax-regret is defined under this assumption). This implies that the end-points of a part should also lie at vertices (since there can be no sink on an edge). For example, if $P_i$ and $P_{i+1}$ are consecutive parts and the right end-point of $P_i$ is some vertex $x_{r_i}$, then the left end-point of $P_{i+1}$ will be $x_{r_i+1}$ and the edge between $x_{r_i}$ and $x_{r_i+1}$ does not take part in the evacuation.

Note that although we are making the sink on vertex assumption, the properties stated later in the paper all hold even without the assumption. With extra complexity, our algorithms could be extended to work without this assumption.


\section{A Solution to the Optimal (Non-Regret) $k$-Sink Location Problem}
\label{Sec:Basic}

\label{sec:opt-k-sinks}

In this section, we wish to show an $O(kn\log n)$ time solution to the optimal $k$-sink location problem. In our problem, the flow (of people) to the sink is \emph{discrete}, since the objects of evacuation (people) are \emph{discrete} entities. The left and right evacuation times to a sink for this \emph{discrete} flow was defined in Eqs. \ref{eq:left-evac-c} and \ref{eq:right-evac-c}. 

We note that \cite{Higashikawaa} has given an $O(kn\log n)$ algorithm for the optimal $k$-sink location problem, but it was for the \emph{continuous} version of the problem (the object of evacuation is \emph{continuous} like a fluid). They derive the following formula for the left evacuation time to a sink $x$ as:
\begin{equation} \Theta_{L}(P,x,s)=\max_{x_i < x}\left\{
(x-x_{i})\tau+ \frac{\sum_{0\leq j\leq i}w_{j}(s)}{c}
\right\} \label{eq:left-evac-continuous}
\end{equation}
and the right evacuation time is defined similarly.

As we can see from Eq. \ref{eq:left-evac-continuous}, the expressions for the \emph{continuous} version does not involve a ceiling function. The method used in \cite{Higashikawaa} is not extendible for the \emph{discrete} case due to complications arising from the ceiling function in the expression.

We observe some key properties of the optimal solution in the \emph{discrete} case and using a novel data structure which we call the \emph{Bi-Heap}, a dynamic programming based procedure has been proposed and explained in detail in Appendix \ref{app:opt-k-sinks-algo}. The running time of the algorithm is $O(kn\log n)$.

\section{Properties of the Minimax-Regret $k$-Sink Location Problem}
\label{Sec:MinMax}

We now elaborate on a few properties of the Minimax-Regret $k$-Sink Location Problem which will help us come up with efficient algorithms.

\subsection{Structure of the \emph{worst-case scenario}}
We note that the adversary $B$ needs to choose a \emph{worst-case scenario} but is free to choose any \emph{worst-case scenario} since all of them give the same \emph{max-regret}. In this section, Lemmas \ref{lem:out-scenario} and \ref{lem:sub-scenario} are going to show that there always exists a \emph{worst-case scenario} with a particular structure and we will assume without loss of generality that $B$'s chosen \emph{worst-case scenario} has this structure. 

Now, assume that the algorithm $A$ has given $k$-partition and sinks $\{\hat{P},\hat{Y}\}$. 

\begin{lemma}
\label{lem:out-scenario}
Let $s^*_B\in\mathcal{S}$ be a \emph{worst-case scenario}  with its dominant part $P_d\in\hat{P}$.
Then  if $s^*_B$ is modified so that  $w_{i}(s^*_B)=w_{i}^{-}$ if $x_i\notin P_d$ (See Fig. \ref{fig:out-scenario}), $s^*_B$ remains a worst case scenario. (Proof in Appendix \ref{app:out-scenario}.)
\end{lemma}


\begin{definition} Given a scenario, a ``\emph{sub-scenario}'' for a subpath is defined as the scenario within that subpath.
\end{definition}

\begin{definition} Consider a subpath $P^{\prime}$ of the path $P$ with leftmost(resp. rightmost) vertex $x_l$(resp. $x_r$). Under some scenario $s\in\mathcal{S}$ a \emph{sub-scenario} for $P^{\prime}$ is called
\emph{left-dominant}(resp. \emph{right-dominant}) if
for some $i$ with $l\leq i\leq r$, $w_{j}(s)=w_{j}^{+}$
(resp. $w_{j}(s)=w_{j}^{-}$) for $l\leq j<i$ and $w_{j}(s)=w_{j}^{-}$
(resp. $w_{j}(s)=w_{j}^{+}$) for $i\leq j\leq r$.
\end{definition}

Given a $k$-partition $\hat{P}=\{P_1,P_2,...,P_k\}$, let $\mathcal{S}_L^i$ (resp. $\mathcal{S}_R^i$) denote the set of all
\emph{left-dominant} (resp. \emph{right-dominant}) sub-scenarios in
the part $P_i$.

\begin{lemma}
\label{lem:sub-scenario} There exists a
\emph{worst-case scenario} $s_B^*$ with its dominant part $P_d\in\hat{P}$ such that the sub-scenario within $P_d$ lies in the set $S_{L}^d\bigcup S_{R}^d$ (See Fig. \ref{fig:sub-scenario}.) (Proof in Appendix \ref{app:sub-scenario}.)
\end{lemma}
\begin{figure}[tp]
\centering
\begin{subfigure}[b]{0.45\textwidth}
\includegraphics[scale=0.4]{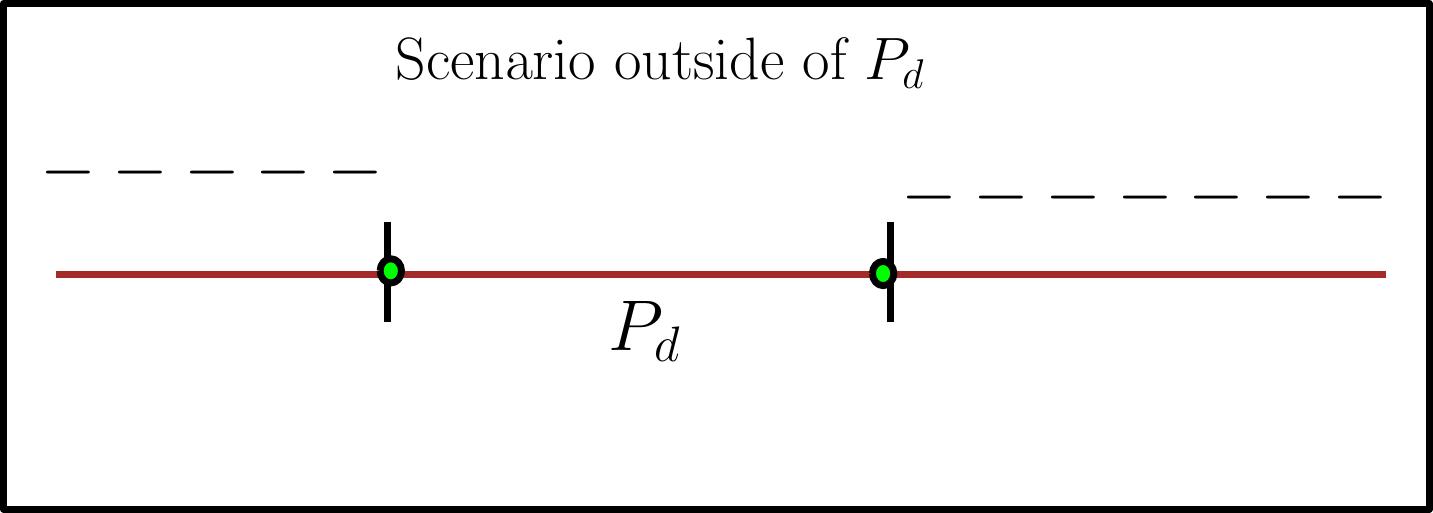}
\caption{Scenario outside $P_d$}
\label{fig:out-scenario}
\end{subfigure}
\hfill
\begin{subfigure}[b]{0.45\textwidth}
\includegraphics[scale=0.4]{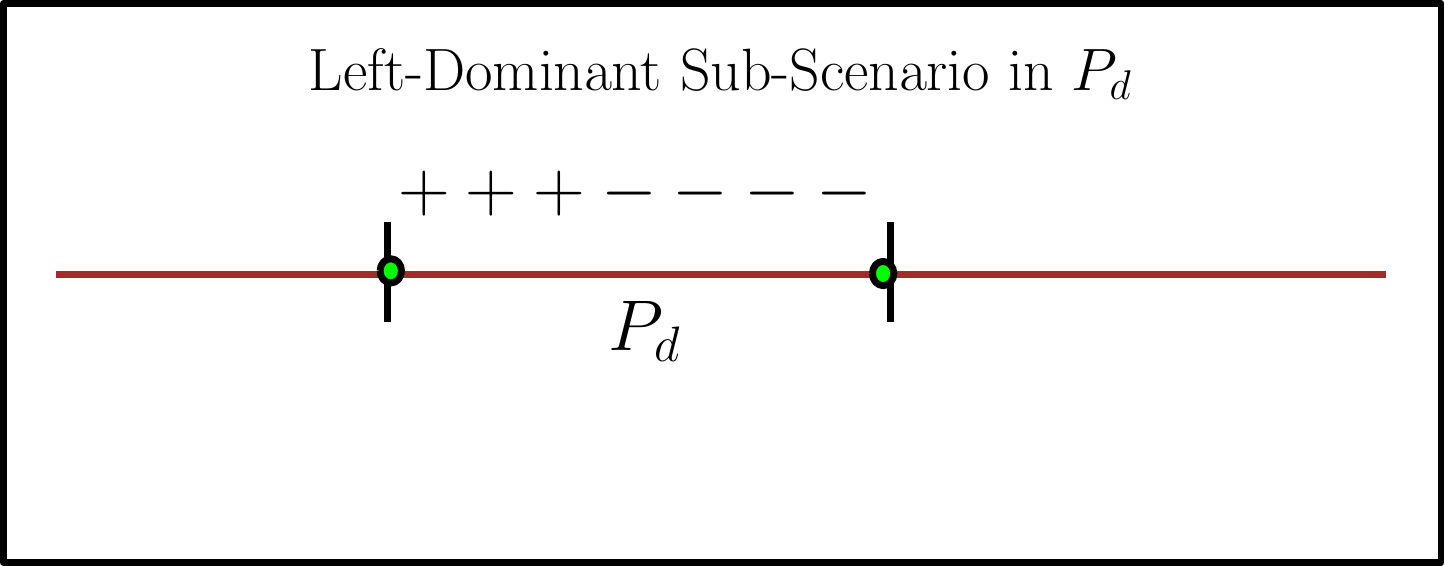}
\caption{A left-dominant sub-scenario in $P_d$}
\label{fig:sub-scenario}
\end{subfigure}
\caption{Structure of a \emph{worst-case scenario}}
\end{figure}

\begin{theorem} \label{thm:wcscenarios}
For any choice of $k$-partition and sinks, $\{\hat{P},\hat{Y}\}$, made by algorithm $A$, there exists a \emph{worst-case scenario} $s^*_B\in\mathcal{S}$ with its dominant part $P_d\in \hat{P}$ of the following form:
\begin{enumerate}
\item $w_{i}(s^*_B)=w_{i}^{-}$ if $x_{i}\notin P_{d}$ and,
\item Sub-scenario in $P_d$ is in the set $\mathcal{S}_{L}^d\bigcup \mathcal{S}_{R}^d$.
\end{enumerate}
\end{theorem}
\begin{proof}
Trivially follows from Lemmas \ref{lem:out-scenario} and \ref{lem:sub-scenario}.\qed
\end{proof}

Let us denote by $\mathcal{S}^*\subseteq \mathcal{S}$, the set of all possible \emph{worst-case scenarios} of the form defined in Theorem \ref{thm:wcscenarios}. Without loss of generality, we assume that the adversary $B$ only chooses a \emph{worst-case scenario} $s_B^*\in\mathcal{S}^*$.

\begin{property}
\label{prop:n-wcscenarios}
Given $k$ parts and sinks $\{\hat{P},\hat{Y}\}$ (or only the parts $\hat{P}$) by algorithm $A$, there are $O(n)$ candidate \emph{worst-case scenarios} from which $B$ chooses $s^*_B\in\mathcal{S}^*$. (Proof in Appendix \ref{app:n-wcscenarios}.)
\end{property}

\begin{property}
\label{prop:n2-wcscenarios}
Overall, there are only $O(n^2)$ possible candidate \emph{worst-case scenarios} for $s^*_B$ irrespective of the $k$-partition/sinks given by the algorithm, i.e., $\left| \mathcal{S}^* \right| = O(n^2)$. (Proof in Appendix \ref{app:n2-wcscenarios}.)
\end{property}

\subsection{Characterization of Minimax-Regret}
Suppose algorithm $A$ chooses $k$-partition and sinks $\{\hat{P},\hat{Y}\}$, we will assume $x_{l_i}$(resp. $x_{r_i}$) to be the left end(resp. right end) of the $i^{th}$ part $P_i\in\hat{P}$. For any scenario $s\in\mathcal{S}$, the regret can be written as:
\begin{align}
R(\{\hat{P},\hat{Y}\},s) & = \Theta^k(P,\{\hat{P},\hat{Y}\},s) - \Theta^k_{\mathrm{opt}}(P,s) && \nonumber \text{(from Eq. \ref{eq:regret})}\\
& =  \max_{1\leq i\leq k}\left\{\Theta^1(P_i,y_i,s)\right\} - \Theta^k_{\mathrm{opt}}(P,s) && \nonumber \text{(from Eq. \ref{eq:k-sink-evac})}\\
& =  \max_{1\leq i\leq k}\left\{\Theta^1(P_i,y_i,s) - \Theta^k_{\mathrm{opt}}(P,s)\right\} && \nonumber \\
& =  \max_{1\leq i\leq k}\left\{R_{l_ir_i}(s,y_i)\right\}. \label{eq:regret1} &&
\end{align}

The term $R_{l_ir_i}(s,y_i)$ refers to the regret under scenario $s\in\mathcal{S}$ when the part $P_i$ is assumed to be the dominant part and $y_i$ is the sink in $P_i$. We refer to $P_i$ as the Assumed Dominant Part(ADP). In reality there is only one unique dominant part, say $P_d$, under scenario $s$. It is easy to see that the maximum is achieved for this dominant part $P_d$, i.e., when $i=d$.

The \emph{max-regret} for $A$'s choice of $\{\hat{P},\hat{Y}\}$ can be written as:
\begin{align}
R_{\max}(\{\hat{P},\hat{Y}\}) & = \max_{s\in \mathcal{S}^*}\left\{R(\{\hat{P},\hat{Y}\},s)\right\} && \text{(from Eq. \ref{eq:max-regret})}\nonumber \\
 & = \max_{s\in S^*}\max_{1\leq i\leq k}\left\{R_{l_ir_i}(s,y_i)\right\} && \text{(from Eq. \ref{eq:regret1})}\nonumber \\
 & = \max_{1\leq i\leq k}\left\{\max_{s\in S^*}\left\{R_{l_ir_i}(s,y_i)\right\}\right\} && \nonumber \\
 & = \max_{1\leq i\leq k}\left\{R_{l_ir_i}(y_i)\right\}. \label{eq:max-regret1}&&
\end{align}

The term $R_{l_ir_i}(y_i)$ denotes the maximum regret over all candidate \emph{worst-case scenarios} if $P_i$ were the ADP with sink $y_i$. This characterization of the \emph{max-regret}, $R_{\max}(\{\hat{P},\hat{Y}\})$, shows the independence of the parts for \emph{max-regret} calculations.

The minimax-regret can be written as:
\begin{align}
\min_{\{\hat{P},\hat{Y}\}} R_{\max}(\{\hat{P},\hat{Y}\}) & = \min_{\{\hat{P},\hat{Y}\}} \max_{1\leq i\leq k}\left\{R_{l_ir_i}(y_i)\right\} && \text{(from Eq. \ref{eq:max-regret1})} \nonumber \\
& = \min_{\hat{P}} \max_{1\leq i\leq k}\left\{\min_{\hat{Y}}R_{l_ir_i}(y_i)\right\} && \label{eq:mmr-it-sinks} \\
& = \min_{\hat{P}} \max_{1\leq i\leq k}\left\{R_{l_ir_i}\right\}. && \label{eq:mmr-rliri} 
\end{align}

The term $R_{l_ir_i}$ denotes the minimax-regret over all candidate \emph{worst-case scenarios} if $P_i$ were the ADP. The minimax-regret for ADP $P_i$ can be thought of as iterating over all possible sinks $y_i \in P_i$ and choosing the sink which minimizes the \emph{max-regret} $R_{l_ir_i}(y_i)$.

We now state two important lemmas which will be used in further proofs.
\begin{lemma} \label{lem:rij-subpath-inc}
For any ADP $P_i$ with left end $x_{l_i}$ and right end $x_{r_i}$, extending the part (subpath) cannot decrease the minimax-regret, i.e., $R_{l_ir_i} \leq R_{l_i(r_i+1)}$ and $R_{l_ir_i} \leq R_{(l_i-1)r_i}$. In other words, extending a subpath cannot decrease the $R_{lr}$ term  for that subpath. (Proof in Appendix \ref{app:rij-subpath-inc}.)
\end{lemma}

\begin{lemma} \label{lem:rijs-subpath-inc}
Let $\hat{P}_q=\{P_1,P_2,...,P_q\}$ be the minimax-regret placement of the first $q$ parts when the right-end of the $q^{th}$ part is restricted to be $x_{r_q}$. Let $\hat{P}_q^{\prime}=\{P_1^{\prime},P_2^{\prime},...,P_q^{\prime}\}$ be the minimax-regret placement of the first $q$ parts when the right-end of the $q^{th}$ part is restricted to be $x_{r_q+1}$. Then, if the left end(resp. right end) of the $i^{th}$ part in $\hat{P}_q$ is $x_{l_i}$(resp. $x_{r_i}$) and in $\hat{P}_q^{\prime}$  is $x_{l_i^{\prime}}$(resp. $x_{r_i^{\prime}}$), then
\begin{equation*}
\max_{1\leq i\leq q}\left\{R_{l_ir_i}\right\} \leq \max_{1\leq i\leq q}\left\{R_{l_i^{\prime} r_i^{\prime}} \right\}
\end{equation*}
In other words, the minimax-regret considering only the first $q$ parts as ADPs for the $q$-partition $\hat{P}_q^{\prime}$ is greater than or equal to that for the $q$-partition $\hat{P}_q$. (Proof in Appendix \ref{app:rijs-subpath-inc}.)
\end{lemma}

\subsection{Unimodality of Minimax-Regret w.r.t Parts}
\begin{property}\label{prop:unimodality-partitions}
In the Minimax-Regret $k$-sink Location Problem, the minimax-regret as a function of the left end of the $i^{th}$ part is unimodal with a unique minimum value given that parts $\{P_{i+1},P_{i+2},...,P_k\}$ are fixed. (Proof in Appendix \ref{app:unimodality-partitions}.)
\end{property} 

\section{A Binary Search Based Algorithm}
\label{Sec:Alg1}

\label{sec:algo-bs}
In this section, we give a binary-search based algorithm for finding the Minimax-Regret $k$-Sink Locations. From Eq. \ref{eq:mmr-rliri}, we know that the $k$-sink minimax-regret can be written as:
\begin{align}
\min_{\{\hat{P},\hat{Y}\}} R_{\max}(\{\hat{P},\hat{Y}\}) & = \min_{\hat{P}} \max_{1\leq i\leq k}\left\{R_{l_ir_i}\right\} && \nonumber
\end{align}

Since we have proved in Property \ref{prop:unimodality-partitions} that the minimax-regret is unimodal as a function of the left end of the $i^{th}$ part when the parts $\{P_{i+1},P_{i+2},...,P_k\}$ are fixed, we can do a nested binary search for the locations.  First, binary search for the left end of the $k^{th}$ part. For each fixed location, binary  search  for the left end of the $(k-1)^{st}$ part, and so on until searching for the left end of the $2^{nd}$ part. This searching will result in $O((\log n)^{k-1})$ possible placements of the parts.

For each fixed placement of the parts, the minimax-regret calculation involves finding the $R_{l_ir_i}$ value for all parts $P_i$, and taking the maximum of the values. This can be done in $O(kn^2\log n)$ time. The details of this procedure can be found in Appendix \ref{app:rliri-procedure}. This procedure to find all $R_{l_ir_i}$ values for fixed placements is called at most $O((\log n)^{k-1})$ times. Therefore, the total running time of the algorithm is $O(kn^2(\log n)^k)$.

This algorithm performs better than the previously known best result of $O\left(n^{1+k} (\log n)^{1 + \log k}\right)$ for general $k$ given in \cite{Ni2014a}.

\section{An Alternative Approach}
\label{Sec:Rec}


\label{sec:alt-approach}
The binary-search based algorithm, though better than previously known results, is not very asymptotically pleasing since there is a $k$ in the exponent. In this section, we are going to give a different approach to solving the minimax-regret $k$-sink location problem. Using the characterization of the minimax-regret given in Eq. \ref{eq:mmr-rliri}, we derive a recurrence for the minimax-regret and implement the recurrence using a Dynamic Programming(DP) approach in Sect. \ref{sec:dp-algo}.

\subsection{A Recurrence for the Minimax-Regret}
From Eq. \ref{eq:mmr-rliri}, the minimax-regret can be written as:
\begin{align}
\min_{\{\hat{P},\hat{Y}\}} R_{\max}(\{\hat{P},\hat{Y}\}) & = \min_{\hat{P}} \max_{1\leq i\leq k}\left\{R_{l_ir_i}\right\} && \nonumber \\
& = \min_{\hat{P}} \max\left(\max_{1\leq i\leq k-1}\left\{R_{l_ir_i}\right\},R_{l_kr_k}\right) && \nonumber \\
& = \min_{P_k} \max\left(\min_{\hat{P}-P_k}\max_{1\leq i\leq k-1}\left\{R_{l_ir_i}\right\},R_{l_kr_k}\right) && \nonumber
\end{align}

The term $\hat{P}-P_k$ denotes the set $\{P_1,P_2,..,P_{k-1}\}$. This gives rise to a natural recurrence for the minimax-reget:
\begin{align}
\min_{\hat{P}} \max_{1\leq i\leq k}\left\{R_{l_ir_i}\right\} & = \min_{P_k} \max\left(\min_{\hat{P}-P_k}\max_{1\leq i\leq k-1}\left\{R_{l_ir_i}\right\},R_{l_kr_k}\right) && \label{eq:mmr}
\end{align}

Let $M(q,i)$ denote the minimax-regret considering only the first $q$ parts as ADPs and the right end of the $q^{th}$ part is $x_i$. From Eq. \ref{eq:mmr}, the following recurrence relation holds for the minimax-regret:
\begin{equation}\label{eq:rec}
M(q,i) = \min_{0\leq j\leq i}\left\{\max\left(M(q-1,j-1),R_{ji}\right)\right\}
\end{equation}

The objective of the algorithm would be to find $M(k,n)$. In Sect. \ref{sec:dp-algo} we implement a DP algorithm based on this recurrence relation.

\section{A Dynamic Programming Based Algorithm}
\label{Sec:Alg2}
\label{sec:dp-algo}

In Sect. \ref{sec:alt-approach}, a recurrence for the minimax-regret was defined in Eq. \ref{eq:rec}. Note that in the recurrence, the index $j-1$ denotes the right end of the $(q-1)^{th}$ part. We will implement this recurrence by first precomputing $R_{ji}$ for all possible values of $j$ and $i$ ($j<i$). Thus $R_{ji}$ for any $j,i$ will be available in a lookup table accessed in $O(1)$ time. This precomputation is described in Sect. \ref{sec:rji-precomp}.

A naive DP table filling implementation of the recurrence in Eq. \ref{eq:rec} is going to give us an $O(kn^2)$ procedure.  This can be improved to $O(kn)$. However, since the DP procedure is not the bottleneck for the algorithm, we defer the details of the $O(kn)$ improvement to Appendix \ref{app:mmr-k-sink-dp}.

\subsection{Precomputation of $R_{ji}$ for all subpaths\label{sec:rji-precomp}}
$R_{ji}$ needs to be precomputed and stored for all subpaths (from $x_j$ to $x_i$) for use in the DP algorithm. Before we precompute $R_{ji}$ for all $j,i$ (for $j<i$) pairs, we need to precompute the optimal $k$-Sink evacuation time for all candidate \emph{worst-case scenarios} $s_B^*\in\mathcal{S}^*$.

From Sect. \ref{sec:opt-k-sinks}, given a scenario we can calculate the $k$-Sink evacuation time in $O(kn\log n)$ time.  Also, from Prop. \ref{prop:n2-wcscenarios}, we know that there are a total of $O(n^2)$ candidate \emph{worst-case scenarios}. Therefore, this precomputation takes time $O(kn^3\log n)$.

\subsubsection{Naive Approach - $O(n^5)$} \label{sec:rji-naive}
\begin{enumerate}
\item There are $O(n^2)$ possible subpaths (Assumed Dominant Parts), say with left-end $x_l$ and right-end $x_r$.
\item Within each subpath, there are $O(n)$ possible sinks $x_t (l\leq t\leq r)$, and 
\item For each sink in the ADP there can be $O(n)$ candidate \emph{worst-case scenarios} $s_B^*\in\mathcal{S}^*$ (From Theorem \ref{thm:wcscenarios}).
\item Given an ADP as a subpath from $x_l$ to $x_r$, a sink $x_t$ and a scenario $s_B^*$, in $O(n)$ time we can calculate the regret $R_{lr}(s_B^*,x_t)$ ($= \Theta^1([x_l,x_r],x_t,s_B^*)  - \Theta^k_{\mathrm{opt}}(P,s_B^*)$), since we have already pre-calculated $\Theta^k_{\mathrm{opt}}(P,s_B^*)$, the optimal $k$-sik evacuation time for all candidate \emph{worst-case scenarios}.
\end{enumerate}

Therefore, the total time complexity to precompute $R_{ji}$ using this approach is $O(n^5)$. At this point, we already have a polynomial time algorithm for the Minimax-Regret $k$-Sink Location Problem, with precomputation of $R_{ji}$'s in $O(n^5)$ time dominating the running time (The precomputation of $\Theta^k_{\mathrm{opt}}(P,s_B^*)$ for all \emph{worst-case scenarios} takes $O(kn^3\log n)$ time and the DP table filling procedure takes $O(kn)$ time.)

However, by observing some properties of the minimax-regret, we are able to bring down the running time of the $R_{ji}$ precomputation to $O(n^3)$. The reduction from $O(n^5)$ to $O(n^3)$ running time is given in Appendix \ref{app:rji-precomputation}.

\subsection{Total Running Time}
The total running time can be split as:
\begin{itemize}
\item Precomputing optimal $k$-sink evacuation times for the $O(n^2)$ candidate \emph{worst-case scenarios} - $O(kn^3\log n)$.
\item Precomputation of $R_{ji}$'s - $O(n^3)$.
\item Algorithm to find minimax-regret k-Sink location - $O(kn)$.
\end{itemize}

The overall running time is dominated by the precomputation of optimal $k$-sink evacuation times for candidate \emph{worst-case scenarios}. Thus, the running time of the algorithm is $O(kn^3\log n)$.

\bibliographystyle{plain} 
\bibliography{main}

\begin{thebibliography}{10}

\bibitem{conf/cocoon/BhattacharyaK12}
Binay~K. Bhattacharya and Tsunehiko Kameda.
\newblock A linear time algorithm for computing minmax regret 1-median on a
  tree.
\newblock In {\em COCOON'2012}, pages 1--12, 2012.

\bibitem{ChengHKNSX13}
Siu-Wing Cheng, Yuya Higashikawa, Naoki Katoh, Guanqun Ni, Bing Su, and Yinfeng
  Xu.
\newblock Minimax regret 1-sink location problems in dynamic path networks.
\newblock In {\em Proceedings of TAMC'2013}, pages 121--132, 2013.

\bibitem{Ford1958a}
L.~R. Ford and D.~R. Fulkerson.
\newblock {Constructing Maximal Dynamic Flows from Static Flows}.
\newblock {\em Operations Research}, 6(3):419--433, June 1958.

\bibitem{Higashikawa2014}
Yuya Higashikawa, M.~J. Golin, and Naoki Katoh.
\newblock {Minimax Regret Sink Location Problem in Dynamic Tree Networks with
  Uniform Capacity}.
\newblock In {\em Proceedings of the 8'th International Workshop on Algorithms
  and Computation (WALCOM'2014)}, pages 125--137, 2014.

\bibitem{Higashikawaa}
Yuya Higashikawa, Mordecai~J Golin, and Naoki Katoh.
\newblock {Multiple Sink Location Problems in Dynamic Path Networks}.
\newblock In {\em Proceedings of the 2014 International Conference on
  Algorithmic Aspects of Information and Management (AAIM 2014)}, 2014.

\bibitem{Hoppe2000b}
B~Hoppe and \'{E} Tardos.
\newblock {The quickest transshipment problem}.
\newblock {\em Mathematics of Operations Research}, 25(1):36--62, 2000.

\bibitem{Kamiyama:2006:EAE:2100322.2100346}
Naoyuki Kamiyama, Naoki Katoh, and Atsushi Takizawa.
\newblock An efficient algorithm for evacuation problems in dynamic network
  flows with uniform arc capacity.
\newblock In {\em Proceedings of the Second international conference on
  Algorithmic Aspects in Information and Management}, pages 231--242, 2006.

\bibitem{Kouvelis1997}
Panos Kouvelis and Gang Yu.
\newblock {\em {Robust Discrete Optimization and Its Applications}}.
\newblock Kluwer Academic Publishers, 1997.

\bibitem{Li2014}
Hongmei Li, Yinfeng Xu, and Guanqun Ni.
\newblock {Minimax regret vertex 2-sink location problem in dynamic path
  networks}.
\newblock {\em Journal of Combinatorial Optimization}, February 2014.

\bibitem{Mamada2006}
Satoko Mamada, Takeaki Uno, Kazuhisa Makino, and Satoru Fujishige.
\newblock {An $O(n \log^2 n) $algorithm for the optimal sink location problem
  in dynamic tree networks}.
\newblock {\em Discrete Applied Mathematics}, 154(2387-2401):251--264, 2006.

\bibitem{Ni2014a}
Guanqun Ni, Yinfeng Xu, and Yucheng Dong.
\newblock {Minimax regret k-sink location problem in Dynamic Path networks}.
\newblock In {\em Proceedings of the 2014 International Conference on
  Algorithmic Aspects of Information and Management (AAIM 2014)}, 2014.

\bibitem{Wang2013}
Haitao Wang.
\newblock {Minmax Regret 1-Facility Location on Uncertain Path Networks}.
\newblock {\em Proceedings of the 24th International Symposium on Algorithms
  and Computation (ISAAC'13)}, pages 733--743, 2013.

\end{thebibliography}

\pagebreak
\section*{\appendixname}
\newcommand{\Cost}[2]{{\rm Cost}(#1,\,#2)}

\begin{appendix}
\section{An $O(kn\log n)$ Dynamic Programming Algorithm for the Optimal $k$-Sink Location Problem}
\label{app:opt-k-sinks-algo}
From this point onwards, whenever we refer to the optimal $k$-sink location problem, we always mean the general \emph{discrete} version, i.e with the ceiling function kept and any capacity $c > 0$ (not necessarily $c=1$ in Eqs. \ref{eq:left-evac-c} and \ref{eq:right-evac-c}.)

We present a solution to the optimal $k$-sink location problem by first establishing a recurrence for the optimal $k$-sink evacuation time and giving a DP algorithm to implement the recurrence. Before doing that, we are going the state the following intuitive lemma with proof deferred to Appendix \ref{app:opt-k-sink-path-subpath}:

\begin{lemma} \label{lem:opt-k-sink-path-subpath}
Given a scenario, the optimal $k$-sink evacuation time of a path is greater than or equal to the optimal $k$-sink evacuation time of any of its subpaths.
\end{lemma}

\subsection{Recurrence for Optimal $k$-Sink Evacuation Time}
We are going to present a recurrence for the optimal $k$-sink evacuation time given a scenario $s\in\mathcal{S}$. Although we prove the recurrence here, an alternative rigorous derivation based on earlier definitions can be found in Appendix \ref{app:evac-recur-derivation}.

Let $T(q,i)$ be the optimal evacuation time considering only the first $q$ parts restricted to the subpath from $x_0$ to $x_i$. Let $w_{ji}$ be the optimal $1$-sink evacuation time in the subpath from $x_j$ to $x_i$. Consider the term:
$$\min_{0\leq j\leq i}\max\left(T(q-1,j-1),w_{ji}\right)$$

For any particular value of $j$, this formula models the $q$-sink evacuation time restricted to the subpath from $x_0$ to $x_i$, with the left end of the $q^{th}$ part being $x_j$, and the first $q-1$ parts being placed optimally for evacuating the subpath from $x_0$ to $x_{j-1}$. The term is minimized over all possible values of $j$, the left end of the $q^{th}$ part.

Since $T(q,i)$ is the {\em optimal} solution,
\begin{equation}\label{eq:k-sink-recur-leq}
T(q,i) \leq \min_{0\leq j\leq i}\max\left(T(q-1,j-1),w_{ji}\right)
\end{equation}

Under scenario $s\in\mathcal{S}$, let us assume that in the optimal $q$-sink placement in the subpath from $x_0$ to $x_i$, the left end of the $q^{th}$ part is $x_{j^{\prime}}$, and the first $q-1$ parts/sinks are given by $\hat{P}_{q-1} = \{P_1,P_2,...,P_{q-1}\}$ and $\hat{Y}_{q-1} = \{y_1,y_2,...,y_{q-1}\}$, the right end of the $(q-1)^{th}$ part being $x_{j^{\prime}-1}$. Let $\Theta^{q-1}([x_0,x_{j^{\prime}-1}],\{\hat{P}_{q-1},\hat{Y}_{q-1}\})$ denotes the $(q-1)$-sink evacuation time under scenario $s$ considering the parts/sinks $\hat{P}_{q-1}$ and $\hat{Y}_{q-1}$ in the subpath from $x_0$ to $x_{j^{\prime}-1}$. Since $T(q-1,j^{\prime}-1)$ is the optimal(minimum) $(q-1)$-sink evacuation time in the subpath from $x_0$ to $x_{j^{\prime}-1}$,
\begin{equation} \label{eq:opt-q-1}
T(q-1,j^{\prime}-1) \leq  \Theta^{q-1}([x_0,x_{j^{\prime}-1}],\{\hat{P}_{q-1},\hat{Y}_{q-1}\})
\end{equation}

Now,
\begin{align}
T(q,i) & = \max\left(\Theta^{q-1}([x_0,x_{j^{\prime}-1}],\{\hat{P}_{q-1},\hat{Y}_{q-1}\}),w_{j^{\prime}i}\right) \nonumber && \\
& \geq \max\left(T(q-1,j^{\prime}-1),w_{j^{\prime}i}\right) \nonumber && \text{(from Eq. \ref{eq:opt-q-1})}\\
& \geq \min_{0\leq j\leq i}\max\left(T(q-1,j-1),w_{ji}\right) \label{eq:k-sink-recur-geq} &&
\end{align}

From Eq. \ref{eq:k-sink-recur-leq} and \ref{eq:k-sink-recur-geq}, we see that 
\begin{equation}\label{eq:k-sink-recur1}
T(q,i) = \min_{0\leq j\leq i}\max\left(T(q-1,j-1),w_{ji}\right)
\end{equation}
where $w_{ji}$ is the optimal $1$-sink evacuation time in the subpath from $x_j$ to $x_i$. There can be multiple values of $j$ (the left end of the $q^{th}$ part) for which $T(q,i)$ is minimized. We will always assume without loss of generality that optimal $j$ is the largest of such values. In other words, the location of the optimal left-end of the $q^{th}$ part is the rightmost possible.

\subsection{A DP Algorithm for Optimal $k$-Sink Location Problem}
In this section we provide a DP algorithm to implement the recurrence given in Eq. \ref{eq:k-sink-recur1}. Note that the DP table for the recurrence will have $O(kn)$ entries. The row index of the DP table corresponds to $q$,  the number of sinks, and the column index $i$ corresponds to $x_i$, the right end-point of the subpath being considered. We fill the table row-by-row so that when we are calculating $T(q,i)$, the values of $T(q-1,j-1)$ (for $0\leq j\leq i$) are already present in the DP table. Again, we note that there can be mutiple values of $j$ that minimize $T(q,i)$ in Eq. \ref{eq:k-sink-recur1} and so, we always choose and refer to the largest possible such value.

We observe two crucial properties of the recurrence which will help us reduce the running time of the DP:

\begin{property} \label{prop:dp-1}
In the recurrence given by Eq. \ref{eq:k-sink-recur1}, if we increment $i$ keeping $q$ fixed, then the minimizing $j$ value cannot decrease.
\end{property}
\begin{proof}
See Appendix \ref{app:prop-dp-1}
\end{proof}

\begin{property} \label{prop:dp-2}
In the recurrence given by Equation \ref{eq:k-sink-recur1} for fixed $q$ and $i$, \\
$\max\left(T(q-1,j-1),w_{ji}\right)$ is unimodal with a unique minimum value as a function of $j$.
\end{property}
\begin{proof}
From Lemma \ref{lem:opt-k-sink-path-subpath}, if we increase(resp. decrease) the size of a path, then the optimal $k$-sink evacuation time cannot decrease(resp. cannot increase). Thus, when we increase $j$, $T(q-1,j-1)$(the optimal $(q-1)$-sink evacuation time on the subpath from $x_0$ to $x_{j-1}$) cannot decrease and $w_{ji}$(the optimal $1$-sink evacuation time on the subpath from $x_j$ to $x_i$) cannot increase. Thus, $\max\left(T(q-1,j-1),w_{ji}\right)$ will have a unique minimum value as a function $j$.\qed
\end{proof}

\begin{remark} \label{rem:faster-dp}
Properties \ref{prop:dp-1} and \ref{prop:dp-2} stated above imply that if we know the minimizing $j$ value for $T(q,i-1)$ (say $j^{\prime}$) then the minimizing $j$ value for $T(q,i)$ can be found by scanning linearly from $j^{\prime}$, stopping when we find the minimum of the unimodal function $\max\left(T(q-1,j-1),w_{ji}\right)$. Note that we will have to overshoot the scanning of $j$ by 1 node in order to determine the minimum of the unimodal function $\max\left(T(q-1,j-1),w_{ji}\right)$.
\end{remark}

The recurrence can be re-written as:
\begin{equation}
T(q,i) = \min_{j^{\prime}\leq j\leq i}\max\left(T(q-1,j-1),w_{ji}\right)
\end{equation}
where $j^{\prime}$ is the minimizing value of $j$ for $T(q,i-1)$.

The recurrence requires the maintenance of $w_{ji}$, the optimal $1$-sink evacuation time in the subpath from $x_j$ to $x_i$ as $i$ and $j$ increases. In Appendix \ref{sec:wji-maintenance}, we introduce a new data structure to maintain $w_{ji}$ as $i$ and $j$ are incremented. Each increment of $i$ or $j$ can be handled in $O(\log n)$ time.

The DP table filling procedure is given below.:
\begin{enumerate}
  \item Fill in $T(1,i) = w_{0i}$ (for $0\leq i\leq n$), the optimal $1$-sink evacuation time in the subpath from $x_0$ to $x_i$. ($O(n\log n)$ time)
  \item For $q\leftarrow 2$ to $k$:
    \begin{itemize}
      \item $j\leftarrow 0$
    \end{itemize}
    \begin{enumerate}
      \item For $i\leftarrow 0$ to $n$:
        \begin{itemize}
          \item Update $w_{ji}$
          \item Do $j\leftarrow j+1$ and update $w_{ji}$ till minimum of $\max(T(q-1,j-1),w_{ji})$ is found which is equal to $T(q,i)$. Note that we will have to overshoot the scanning of $j$ by 1 in order to determine the minima.
          \item Store $T(q,i)$ in the DP table along with the minimizing $j$ value and the optimal sink for subpath $[x_j,x_i]$ (used to reconstruct the optimal $k$-partition and sinks.)
        \end{itemize}
    \end{enumerate}
\end{enumerate}

As we can see by the DP procedure, for every $q$ value, $i$ is incremented at most $n$ times and $j$ at most $2n$ times (because of the overshooting by 1). Thus, $w_{ij}$ is updated atmost $O(n)$ times for every $q$ value and every update incurs a cost of $O(\log n)$ to maintain $w_{ji}$. Thus, we are able to fill the $O(kn)$ DP table entries in $O(kn\log n)$ time.

\subsection{Maintenance of $w_{ji}$}
\label{sec:wji-maintenance}
In the DP algorithm, we also need to maintain $w_{ji}$, the optimal $1$-sink evacuation time, as $j$ and $i$ increases. Therefore any data structure which maintains $w_{ji}$ must be able to maintain the maximum of the \emph{evacuation function} of every node defined in the Eqs. \ref{eq:left-evac-c} and \ref{eq:right-evac-c}, and support the following operations:
\begin{enumerate}
\item Return the optimal $1$-sink evacuation time in subpath $\{x_j,...,x_i\}$.
\item When $i\leftarrow i+1$, then append a new node to the right of the path. 
\item When $j\leftarrow j+1$, then delete the leftmost node from the path.
\end{enumerate}

We introduce a new data structure which we call a \emph{Bi-Heap} which can maintain $w_{ji}$ and perform all operations in $O(\log n)$ time. For details of the data structure, refer Appendix \ref{app:bi-heap}.

\subsection{Running time}
We are filling $O(kn)$ DP entries and each takes $O(\log n)$ time (for maintaining $w_{ji}$.) Thus the overall running time of the algorithm is $O(kn\log n)$.

\section{Proofs for the Optimal $k$-Sink Evacuation}
\subsection{Proof of Lemma \ref{lem:opt-k-sink-path-subpath}}
\label{app:opt-k-sink-path-subpath}
Assume a path $P$ and a subpath $P^{\prime}$ of the path $P$. Now, given a scenario, if the optimal $k$-sink evacuation time in $P^{\prime}$ is greater than the optimal $k$-sink evacuation time in $P$, then we can use the $k$-partition used in the optimal solution for $P$ on $P^{\prime}$ ($\leq k$ parts) and obtain a better solution than the optimal one in $P^{\prime}$. Therefore, the optimal $k$-sink evacuation time in $P$ is greater than or equal to the optimal $k$-sink evacuation time in any of its subpaths $P^{\prime}$.\qed

\subsection{A Derivation for the Optimal $k$-Sink Evacuation  Recurrence}
\label{app:evac-recur-derivation}
The $k$-sink evacuation time is given by the following equation:
\begin{align}
\Theta^k(P,\{\hat{P},\hat{Y}\},s) & = \max_{1\leq i\leq k} \Theta^1(P_{i},y_i,s) \nonumber && \text{(from Eq. \ref{eq:k-sink-evac})}\\
& = \max\left(\max_{1\leq i\leq k-1} \left\{\Theta^1(P_{i},y_i,s)\right\},\Theta^1(P_k,y_k,s)\right) && \\
& = \max\left(\Theta^{k-1}(P-P_k,\{\hat{P}-P_k,\hat{Y}-y_k\},s),\Theta^1(P_k,y_k,s)\right) \label {eq:k-sink-evac1} &&
\end{align}
where $P-P_k$ denotes the path excluding the $k^{th}$ part and $\hat{P}-P_k$ denotes the set $\{P_1,P_2,...,P_{k-1}\}$.

Thus the optimal $k$-sink evacuation time can be written as:
\begin{align}
\Theta^k_{\mathrm{opt}}(P,s) & = \min_{\{\hat{P},\hat{Y}\}} \Theta^k(P,\{\hat{P},\hat{Y}\},s) \nonumber && \\
& = \min_{\{\hat{P},\hat{Y}\}} \max\left(\Theta^{k-1}(P-P_k,\{\hat{P}-P_k,\hat{Y}-y_k\},s),\Theta^1(P_k,y_k,s)\right) \nonumber && \text{(from Eq. \ref{eq:k-sink-evac1})}\\
& = \min_{\{P_k,y_k\}} \max\left(\min_{\{\hat{P}-P_k,\hat{Y}-y_k\}}\Theta^{k-1}(P-P_k,\{\hat{P}-P_k,\hat{Y}-y_k\},s),\Theta^1(P_k,y_k,s)\right) \nonumber && \\
& = \min_{\{P_k,y_k\}} \max\left(\Theta^{k-1}_{\mathrm{opt}}(P-P_k,s),\Theta^1(P_k,y_k,s)\right) \nonumber && \\
& = \min_{P_k} \max\left(\Theta^{k-1}_{\mathrm{opt}}(P-P_k,s),\min_{y_k}\Theta^1(P_k,y_k,s)\right) \nonumber && \\
& = \min_{P_k} \max\left(\Theta^{k-1}_{\mathrm{opt}}(P-P_k,s),\Theta^1_{\mathrm{opt}}(P_k,s)\right) \label{eq:k-sink-recur} &&
\end{align}

As we can see from Eq. \ref{eq:k-sink-recur}, a natural recurrence relation exists for the optimal $k$-sink evacuation time.

Let $T(q,i)$ be evacuation time considering only the first $q$ parts and only in the subpath from $x_0$ to $x_i$. From Eq. \ref{eq:k-sink-recur}, we can write $T(q,i)$ as:
\begin{equation*}
T(q,i) = \min_{0\leq j\leq i}\max\left(T(q-1,j-1),w_{ji}\right)
\end{equation*}
where $w_{ji}$ is the optimal $1$-sink evacuation time in the subpath from $x_j$ to $x_i$. There can be multiple values of $j$ (the left end of the last part) for which $T(q,i)$ is minimized. We assume w.l.g that the largest value of minimizing $j$ is always chosen and meant when we say minimizing $j$ value.

\subsection{Proof of Property \ref{prop:dp-1}}
\label{app:prop-dp-1}

Equation \ref{eq:k-sink-recur1} is given by:
\begin{equation}
T(q,i) = \min_{0\leq j\leq i}\max\left(T(q-1,j-1),w_{ji}\right)
\end{equation}
where $w_{ji}$ is the optimal $1$-sink evacuation time in the subpath from $x_j$ to $x_i$.

Assume $j^{\prime}$ is the  minimizing $j$ value for $T(q,i-1)$ and $j^{\prime\prime}$ is the minimizing $j$ value for $T(q,i)$. We need to prove that $j^{\prime\prime}\geq j^{\prime}$.

Assume the contrary, i.e., $j^{\prime\prime}<j^{\prime}$. We now have two cases to deal with:\\\\
\textbf{Case 1}: $w_{j^{\prime}(i-1)} \geq T(q-1,j^{\prime}-1)$

In this case, the $q$-sink evacuation time in the subpath from $x_0$ to $x_{i-1}$ will be $T(q,i-1) = w_{j^{\prime}(i-1)}$. Now,
\begin{equation*}
T(q,i) = \max\left(T(q-1,j^{\prime\prime}-1),w_{j^{\prime\prime}i}\right)
\end{equation*}

\begin{remark} \label{rem:path-subpath-evac-time}
The optimal $k$-sink evacuation time of a path is greater than or equal to the $k$-sink evacuation time of any of its subpaths. (From Lemma \ref{lem:opt-k-sink-path-subpath})\end{remark}

Since $j^{\prime\prime} < j^{\prime}$, 
\begin{align}
T(q-1,j^{\prime\prime}-1) & \leq T(q-1,j^{\prime}-1) \nonumber && \text{(from Remark \ref{rem:path-subpath-evac-time})}\\
 & \leq w_{j^{\prime}(i-1)} \nonumber && \text{(from case assumption)} \\
 & \leq w_{j^{\prime}i} \nonumber && \text{(from Remark \ref{rem:path-subpath-evac-time})}\\
& \leq w_{j^{\prime\prime}i} && \text{(from Remark \ref{rem:path-subpath-evac-time})} \label{eq:prop-1-1}
\end{align}

Therefore from Eq. \ref{eq:prop-1-1}, $T(q,i)=w_{j^{\prime\prime}i}$. Since $j^{\prime\prime}$ is the largest minimizing $j$ value for $T(q,i)$ and $j^{\prime} > j^{\prime\prime}$,
\begin{align}
w_{j^{\prime\prime}i} & < \max\left(T(q-1,j^{\prime}-1),w_{j^{\prime}i}\right) && \label{eq:prop-1-2} \\
& = w_{j^{\prime}i}
\end{align}
But from Remark \ref{rem:path-subpath-evac-time}, we know that the optimal $1$-sink evacuation time of the path $x_{j^{\prime\prime}}...x_i$ cannot be strictly lesser than the $1$-sink evacuation time of one of its subpaths $x_{j^{\prime}}...x_i$. Eq. \ref{eq:prop-1-2} cannot be true. We arrive at a contradiction. Therefore, our assumption that $j^{\prime\prime} < j^{\prime}$ is false for this case. So,  $j^{\prime\prime}\geq j^{\prime}$.\\\\
\textbf{Case 2}: $w_{j^{\prime}(i-1)} < T(q-1,j^{\prime}-1)$

In this case, the $q$-sink evacuation time in the subpath from $x_0$ to $x_{i-1}$ will be $T(q,i-1) = T(q-1,j^{\prime}-1)$. The minimizing $j$ value for $T(q,i)$ is $j^{\prime\prime}$,
\begin{equation*}
T(q,i) = \max\left(T(q-1,j^{\prime\prime}-1),w_{j^{\prime\prime}i}\right)
\end{equation*}

Since  $j^{\prime\prime}$ is the largest minimizing $j$ value for $T(q,i)$, and $j^{\prime} > j^{\prime\prime}$,
\begin{equation} \label{eq:prop-1-3}
\max\left(T(q-1,j^{\prime\prime}-1),w_{j^{\prime\prime}i}\right) < \max\left(T(q-1,j^{\prime}-1),w_{j^{\prime}i}\right)
\end{equation}

By the case assumption and the fact that $j^{\prime}$ is the minimizing $j$ value for $T(q,i-1)$,
\begin{align}
w_{j^{\prime}(i-1)} & < T(q-1,j^{\prime}-1) \nonumber && \text{(from case assumption)}\\
& \leq \max(T(q-1,j^{\prime\prime}-1),w_{j^{\prime\prime}(i-1)}) \nonumber && \\
& \leq \max(T(q-1,j^{\prime\prime}-1),w_{j^{\prime\prime}i}) \label{eq:prop-1-4} && \text{(from Remark \ref{rem:path-subpath-evac-time})}
\end{align}
Therefore,
\begin{align}
T(q-1,j^{\prime}-1) & \leq \max(T(q-1,j^{\prime\prime}-1),w_{j^{\prime\prime}i}) \label{eq:prop-1-5} && \text{(from Eq. \ref{eq:prop-1-4})}\\
w_{j^{\prime}i} & \leq w_{j^{\prime\prime}i} && \label{eq:prop-1-6} \text{(from Remark \ref{rem:path-subpath-evac-time})}
\end{align}

From Eqs. \ref{eq:prop-1-5} and \ref{eq:prop-1-6},
\begin{equation}\label{eq:prop-1-7}
\max\left(T(q-1,j^{\prime}-1),w_{j^{\prime}i}\right) \leq \max(T(q-1,j^{\prime\prime}-1),w_{j^{\prime\prime}i})
\end{equation}

We see that Eqs. \ref{eq:prop-1-3} and \ref{eq:prop-1-7} contradict each other. Therefore, our assumption that $j^{\prime\prime} < j^{\prime}$ is false for this case. So, $j^{\prime\prime} \geq j^{\prime}$. \qed

\section{The $Bi-Heap$ Data Structure}
\label{app:bi-heap}
We introduce this data structure to help maintain $w_{ji}$, the optimal $1$-sink evacuation time in the subpath from $x_j$ to $x_i$ as $j$ and $i$ are incremented in $O(\log n)$ time. This maintenance can be done efficiently by combining together various basic data structues, e.g., heaps and 2-3 trees.  Whle not difficult, the details are quite technical.

\begin{remark} \label{rem:1-sink-evac-opt}
In the subpath from $x_j$ to $x_i$, the position of the optimal $1$-sink cannot move left if $j$ and/or $i$ are incremented.
\end{remark}

\subsection{Setup}
For any node $x_i (0\leq i < n)$, set $\ell_i=x_{i+1}-x_i$ to be the distance from node $i$ to $i+1$  and $w_i>0$ be the number of people at node $i$ who need to be evacuated.

Let $\ell_{i,j} = x_j - x_i$ and $W_{i,j} =\sum_{i \le t \le j} w_{t}$. Suppose all edges have fixed uniform capacity $c.$ For $i < j$ the time to evacuate the nodes in $[i,j]$ to node $j$ is  $E_{i,j} - 1$ and the time to evacuate the nodes in $[i,j]$ to node $j$ is  $E'_{i,j} - 1$ where 

\begin{equation}
E_{i,j} = \max_{i \le t < j} 
\left(
 \left \lceil \frac {W_{i,t}} c  \right \rceil  + \ell_{t,j}  \right),
 \quad \quad
 E'_{i,j} = \max_{i <  t \le j} 
\left(
 \left \lceil \frac {W_{t,j}} c  \right \rceil  + \ell_{i,t}  \right) 
\end{equation}

In order to maintain the optimal $1$-sink evacuation time for the last part, we maintain the left evacuation time $E_{i,j}$ and the right evacuation time $E^{\prime}_{i,j}$ for the optimal sink. Note that by Remark \ref{rem:1-sink-evac-opt}, the optimal sink in the last part cannot move to the left. This implies that we must be able to handle the updates from $E_{i,j}$ to $E_{i+1,j}$ or $E_{i,j+1}$ and from $E^{\prime}_{i,j}$ to $E^{\prime}_{i+1,j}$ or $E^{\prime}_{i,j+1}$.

One method of handling these updates is doing a $O(n^2)$ preprocessing step that precalculates {\em all} possible $E_{i,j}$ and $E'_{i,j}$ values for later use, using $O(n^2)$ space.

Another approach, the one that will be outlined here, is to create a $O(n)$ size data structure that permits reusing old information about $E_{i,j}$ or $E'_{i,j}$ to calculate the value of
$E_{i+1,j}$,  $E_{i,j+1},$,  $E'_{i+1,j}$,  or $E'_{i,j+1}$ in $O(\log n)$ time.

\subsection{Some observations}

Suppose $i < j$ are given and  $j$ is increased by $1$.  Note that
\begin{eqnarray*}
E_{i,j+1} &=& \max_{i \le t < j+1} 
\left(
 \left \lceil \frac {W_{i,t}} c  \right \rceil  + \ell_{t,j+1}  
 \right)\\
 &=& \max \left(
 \max_{i \le t < j} 
\left(
 \left \lceil \frac {W_{i,t}} c  \right \rceil  + \ell_{t,j}  + \ell_j
 \right),
 \left \lceil \frac {w_j} c \right \rceil + \ell_j
\right)\\
\end{eqnarray*}

That is, the maximum is taken over {\em almost} the same set of items as $E_{i,j}$ except that $\ell_j$ is added to all of the old items and one new item is added.

If $i$ is increased by one then 

\begin{eqnarray*}
E_{i+1,j} &=& \max_{i+1 \le t < j} 
\left(
 \left \lceil \frac {W_{i+1,t}} c  \right \rceil  + \ell_{t,j}  
 \right)\\
 &=&  \max_{i+1 \le t < j} 
\left(
 \left \lceil \frac {W_{i,t}-w_i} c  \right \rceil  + \ell_{t,j} 
 \right)
\end{eqnarray*}

The maximum here is again taken over {\em almost} the same set of items as $E_{i,j}$ except that one item is removed from the set and $w_i$ is subtracted from all of the $W_t$.

We now examine $E'_{i,j}$ and find that the situation is very similar.

If $j$ is increased by $1$ then
\begin{eqnarray*}
 E'_{i,j+1} &=& \max_{i <  t \le j+1} 
\left(
 \left \lceil \frac {W_{t,j+1}} c  \right \rceil  + \ell_{i,t}  \right)\\
 &=& \max \left(
 \max_{i < t \le j} 
\left(
 \left \lceil \frac {W_{t,j}+w_{j+1}} c  \right \rceil  + \ell_{i,t}
 \right),
 \left \lceil \frac {w_{j+1}} c \right \rceil + \ell_{i,j+1}
\right)\\
\end{eqnarray*}

The maximum is taken over {\em almost} the same set of items as $E'_{i,j}$ except that $w_{j+1}$ is added to all of the old items and one new item is added.

If $i$ increases by 1 we get
\begin{eqnarray*}
 E'_{i+1,j} &=& \max_{i+1 <  t \le j} 
\left(
 \left \lceil \frac {W_{t,j}} c  \right \rceil  + \ell_{i+1,t}  \right)\\	
 &=& \max_{i+1 <  t \le j} 
\left(
 \left \lceil \frac {W_{t,j}} c  \right \rceil  + \ell_{i,t} - \ell_{i+1}  \right)
\end{eqnarray*}
where the maximum is now taken over {\em almost} the same set of items as $E'_{i,j}$ except that one item is removed from the set and $\ell_{i+1}$ is subtracted from all of the $\ell_{i,t}.$

\subsection{The {\em Bi-Heap} Data Structure}
In order to simplify the exposition we define.
$$\Cost{W}{L} = \left\lceil \frac {W} c \right\rceil + L.$$

Following the observations in the previous section we note that to solve the problem it suffices to implement a data structure that maintains a set of pairs
$Z = \{(W_i,L_i)\}$, allowing the following operations on the set, with each operation requiring only $O(\log n)$ time:
\begin{enumerate}
\item $MAX$: returns $\max \left\{ \Cost{W_i}{L_i}\, :\, (W_i,L_i) \in Z\right\}$
\item  $AddW(w):$ In  every $(W_i,L_i) \in Z$, replace $W_i$ by $W_i +w$.
\item  $AddL(\ell):$ In every $(W_i,L_i) \in Z$, replace $L_i$ by $L_i +\ell$.
\item  $Insert((W,L)):$ Insert new pair $(W,L)$ into $Z$
\item  $Delete((W,L):$ Remove pair $(W,L)$ from $Z$.  Note that the input here will be a pointer to the current location of $(W,L)$ in $Z$.
\end{enumerate}

We will call such a data  structure a 
{\em Bi-Heap}.
$E_{i,j}$ can be evaluated by maintaining a Bi-Heap on the appropriate $(W,L)$ pairs and calculating $E_{i,j}$ via the MAX function.  $E_{i,j}$ can be updated to  $E_{i,j+1}$ in 
$O(\log n)$ time by first performing a $AddL(\ell_j)$ operation and then an $Insert((w_j,\ell_j))$ one. 
The value of $E_{i,j+1}$ is then found by calling $MAX$.
 A similar observation permits updating in $O(\log n)$ time  from $E_{i,j}$ to $E_{i+1,j}$, 
 and from $E'_{i,j}$ to
 $E'_{i,j+1}$ or $E'_{i+1,j}.$
 
We will now see how to implement  such a data structure.

At time $t$ define $\bar w_t = w$ if the current operation is $AddW(w)$;  define $\bar \ell_t = \ell$
if the current operation is $AddL(\ell).$  Otherwise, set $\bar w_t = 0$ and $\bar \ell_t = 0.$
Define $\bar W_t = \sum_{t' \le t} \bar w_{t'}$ and $\bar L_t = \sum_{t'\le t} \bar \ell_{t'}.$
Our algorithm will keep the current value of $\bar W_t$ in a variable $\bar W$ and the current value of
$\bar L_t$ in a variable $\bar L.$ This can be maintained in $O(1)$ time per step.

The main issue in designing the data structure is dealing with the $\lceil W/ c\rceil$ term in the cost function.
Observe that $\Cost{W_i}{L_i+\ell}= \Cost{W_i}{L_i} + \ell$. Thus $AddL(\ell)$ does not change the  
relative ordering of the items in $Z$.
If $\Cost{W}{L}$ was actually defined by $\frac W c + L$ then $\Cost{W_i+w}{L}= \Cost{W_i}{L_i} + \frac w c$ 
so $AddW(w)$ would also not change the  relative ordering of the items in $Z$.  This means that $Z$ could be maintained by a priority queue in $O(\log n)$ time per operation.  The only subtle point  is that the  priority queue would {\em not} store the actual current values $(W_i,L_i)$ in the data structure.
Let $t'$ be the time at which item $(W,L)$ was inserted. The priority queue will store $W,L$ {\em and}
the associated values    $\bar W_{t'}, \bar N_{t'}$. The value of the pair in the pair in the 
priority queue at the current time is
$$(W+\bar W - \bar W_{t'},  L + \bar L - \bar L_{t'})$$
which  can be calculated on the fly in $O(1)$ time.

\medskip

The reason this idea does not fully work for the actual problem is  that with the real definition
$$\Cost{W}{L} = \left\lceil \frac {W} c \right\rceil + L$$
the relative ordering of the cost function is {\em not} maintained under $AddW(w)$.  More specifically, it is not difficult to find examples when $\Cost{W}{L} < \Cost{W'}{L'}$ but $\Cost{W+w}{L} > \Cost{W'+w}{L'}$.  So, $AddW(w)$ operations can change the relative ordering of the items in $Z$ under Cost. It is this complication that requires augmenting the data structure.

This issue can be  addressed by first  partitioning  the $(W_i,L_i)$ pairs into (at most) $c$ subsets $Z_d,$ $0 \le d < c$ where
$$Z_d = \left\{(W_i,L_i) \in Z \,:\,  W_i \equiv d \bmod c \right\} $$

Note that 
$$\max \left\{  \Cost{W_i}{L_i}\, :\, (W_i,L_i) \in Z\right\}
= \max_{0 \le d < c}  R_d 
$$
where
$$ R_d = \left\{ \Cost{W_i} {L_i}\, :\, (W_i,L_i) \in Z_d\right\}.$$
The idea is, to identify, for each $Z_d$,  the item with largest cost, $R_d$, and only maintain  the maximum over those representative values and not over the full set of items. These $R_d$ will be kept in an auxiliary max-Heap $H$.

Next note that, if $W \equiv W' \bmod c$, then 
$$\left \lceil 
\frac {W + w} c 
\right \rceil 
-
\left \lceil 
\frac {W} c 
\right \rceil 
= 
\left \lceil 
\frac {W' + w} c 
\right \rceil
-
\left \lceil 
\frac {W'} c 
\right \rceil.
$$
So, by definition, for fixed $d$,  the items in a fixed $Z_d$ stay in the same cost sorted order after $AddW$ and $AddL$ operations.  Thus, the representative of $Z_d$ remains the same after those operations.
 
Our algorithm will keep each set $Z_d$ in its own max-heap. Note that there is some ambiguity in the labelling of  $Z_d.$  Suppose that$(W,L) \in Z_d$.
After an $AddW(w)$ operation $W$ is replaced by $W + w$ so the pair is now in $Z_{d'}$ where
$$d' = (W+w) \bmod c = (d+ w) \bmod c.$$  
We therefore label  the Max-Heaps  by the $d$ values they would have at time $t=0.$  More specifically,  we define $H_0,H_1,\ldots, H_{c-1}$, only explicitly storing the non-empty ones. At time $t$, $Z_d$ will be stored in
$H_{d'}$ where $d' + \bar W_t \equiv d \bmod c.$

We now define the pieces of our data structure

\begin{enumerate}
\item $\bar W, \bar L$.  These values are set to the current values $\bar W  = \bar W_t$, and $\bar L =  \bar L_t:$ 
\item Max heaps  $H_0,H_1,\ldots, H_{c-1}$ storing  the subsets $Z_d$
	\begin{enumerate}
			\item Implemented using a binary-tree heap data structure ordered using 
		    cost  $\left\lceil  {W_i} /c \right\rceil + L_i$. Allows identifying the max value, insertion and deletion of an element in $O(\log n)$ time, where $n$ is the number of elements in the heap.		
	\item Only non-empty $Z_d$ would be stored with one max-heap for each nonempty $Z_d$
		\item $Z_d$ is stored in Heap $H_{d'}$ where $d'= (d - \bar W_t) \bmod c$.
		\item Each  entry in a  $H_{d'}$  contains the $(W,L)$ values of the pair 
			{\em at the time of insertion}
		  			as well as the values $\bar W_t, \bar L_t$ {\em at the time of insertion}.
		\item Identifies the largest cost item in $Z_d$ as $R_d$.  This will also be labelled as $r_{d'}$
	\end{enumerate}
\item A dictionary $D$ allowing access to the existing $H_{d'}$  heaps by their labels.  
	\begin{enumerate}
		\item implemented using a binary tree
	\end{enumerate}
\item A special type of max  Heap  $H$ that stores the (maximum cost) representative from each non-empty $Z_d$.
		\begin{enumerate}
			\item Stores the representatives, $r_{d'}$  of the non-empty $H_{d'}$ heaps
			\item Representatives $r_{d'}$ are stored with their index $d'$ and a pointer to the entry 
				they represent in $H_{d'}$. 
			\item Permits calculating MAX (of the representatives) in $O(1)$ time and  
			inserting and deleting items from $H$ in $O(\log n)$ time, where $n$ is the current number of items
			 in $H$.
			\item Has one more special operation to be described below.
		\end{enumerate}
\end{enumerate}

$H$ will be implemented by using a $2$-$3$ tree with the leaves of the tree being the 
$r_{d'}$. In a $2$-$3$ tree, all leaves are on the same depth 
 so the height of the tree is $O(\log n).$   

 The actual cost of a leaf  is not stored (since it is always changing) but can be calculated in $O(1)$ time from its index $d'.$

The  leaves of $H$ are sorted from left to right by increasing index $d'$. The leaf descendants
of any internal node $v$ therefore form a continuous interval of leaf nodes.  We let $I(v)$ denote
this interval.

Internal node $v$  contain two pieces of information in addition to pointers to its (2 or 3) children and its parent: the {\em left endpoint} which is the smallest index of a representative in $I(v)$  and the {\em max cost index} which is the index of the largest cost item in $I(v)$.  The information in an internal node  can therefore be calculated from the information in its (2 or 3) children; the left endpoint  is the smallest of the left endpoints  of  its (2 or 3) children and  the max cost node is the max cost index of its (2 or 3) children that has the largest cost.
 Note that the max cost index in the root $r$ is the max cost node in $I(r)$ and is therefore
 the max cost item in $Z$. 

Considering $H$ as a max heap on its leaves, standard techniques permit inserting a new leaf or deleting an old one in $O(n)$ time.
Furthermore, note that if ALL the values of all of the leaves in the tree are changed by the same amount $\ell$ then the tree does not change at all.  
\medskip

We can now complete the description of the Bi-Heap and its operations.

\medskip

\par\noindent\underline{\bf MAX:} Simply report the max cost node in $H$.

\medskip

\par\noindent\underline{\bf  $ \bf Insert(W,L)$ or $\bf Delete(W,L)$:}\\

Before doing anything else  first calculate $d = (W \bmod C)$ to identify the $Z_d$ to  which $(W,L)$ belongs and $d'= (d - \bar W_t) \bmod c$
to identify the $H_{d'}$ which stores $Z_d.$

For $Insert$, if $Z_d$ is empty and no such $H_{d'}$ exists, create it and insert $H_{d'}$  into $D$.  

Now insert $(W,L)$ into $H_{d'}$.
If this new item is  {\em not} the largest cost item in $H_{d'}$, stop.  
If it is, then delete  the current $d'$-representative from $H$ and Insert $(W,L)$ into $H$ as the newest $d'$-representative.

For $Delete$ first check if $(W,L)$ is the largest item in $H_{d'}$.  If it is not, then just delete it from $H_{d'}$ and stop. If $H_{d'}$ is now empty, delete it from $D$ 

If $(W,L)$ is the largest item in $H_{d'}$  then it is also the $d'$-representative in $H$.  Delete it from $H$ and $H_{d'}$.  
If $H_{d'}$ is now empty, delete it from $D$ If $H_{d'}$ is not empty then insert the largest remaining item in $H_{d'}$ into $H$ as the $d'$-representative.

\medskip
\par\noindent\underline{\bf $\bf AddL(\ell)$:}

 $AddL(\ell)$ operations only change the actual values of the items but not their relative orderings.  So, nothing needs to be changed. The change in $\ell$ value is taken care of by 
 setting $\bar L = \bar L + \ell.$
 
\medskip
\par\noindent\underline{\bf $\bf AddW(w)$:}
First set $\bar W = \bar W + w.$.

We assume $w > 0$  ($w < 0$ is very similar). 
Note  that for each fixed $d$, the ordering of the items in $Z_d$  remains invariant  after the operation, so the $H_d'$ structures do not change.  The only thing that might change is 
the  relative orderings of the $r_d'$-representatives in $H$.
 
First assume $w \equiv 0 \bmod c$.  Then, for all $(W,L)$, $Cost(W+w,L) = Cost(W,L) + w/c.$  Thus the relative ordering of the $d'$ representatives 
in $H$ remains invariant and nothing further needs to be done.

\medskip

Otherwise, $w = \bar k c + \bar d$ where $\bar k, \bar d$  are  integers  and $\bar d > 0.$  Thus $\lceil w/c \rceil = \bar k+1.$

Now let $(W,L) \in Z_d$.  Then $W = k c + d$ and 
$$ Cost(W,L)=
\left\{
\begin{array}{ll}
k+ L,  &  \mbox{if } d =0\\
k + L + 1,& \mbox{if } d \not = 0
\end{array}
\right..
$$
Thus, if $W \in Z_d$, then
\begin{eqnarray}
Cost(W+w,L)
&=& \left\lceil \frac {kc + \bar k c + d + \bar d} c \right\rceil + L\\
&=&\left\{
\begin{array}{ll}
k+ L + \bar k +1  &  \mbox{if } d+\bar d \le c\\[0.3in]
k + L + 2,& \mbox{if } d + \bar d > c
\end{array}
\right.\\[0.2in]
&=&\left\{
\begin{array}{ll}
Cost(W,L) + \lceil w/c \rceil,   &  \mbox{if } d=0\\
Cost(W,L) + \lceil w/c \rceil  - 1,& \mbox{if } d\not= 0,\, d + \bar d \le c\\
Cost(W,L) + \lceil w/c \rceil, & \mbox{if } d+\bar d > c
\end{array}
\right.\\[0.2in]
&=&\left\{
\begin{array}{ll}
Cost(W,L) + \lceil w/c \rceil -1,   &  \mbox{if }  0 < d \le c - \bar d\\
Cost(W,L) + \lceil w/c \rceil ,& \mbox{otherwise}\\
\end{array}
\right. \label{drange}
\end{eqnarray}

For algorithmic purposes we need to transform $d$ into $d'.$  
Recall that $d = d' + \bar W_t \bmod c$.  Set $x = \bar W_t \bmod c$.  Then $d = d' + x \bmod c$.

There are three cases.  The first is that $x =0.$  In this case $d'=d$ and 
(\ref{drange}) stays the same after substituting $d'$ for $d$.

In the second  $c- \bar d \le x < c$ or equivalently,  $0 < c-x \le \bar d$.  
(\ref{drange}) can now be rewritten as 
\begin{equation}
Cost(W+w,L)
=
\left\{
\begin{array}{ll}
Cost(W,L) + \lceil w/c \rceil -1,   &  \mbox{if }  c-x < d' \le c - \bar d + c -x\\
Cost(W,L) + \lceil w/c \rceil,& \mbox{otherwise}\\
\end{array}
\right.
\end{equation}
The final case is  $0 < x  <c - \bar d$, for which (\ref{drange})  can be rewritten as

\begin{equation}
Cost(W+w,L)
=
\left\{
\begin{array}{ll}
Cost(W,L) + \lceil w/c \rceil,   &  \mbox{if }  c-\bar d - x < d' \le c -x\\
Cost(W,L) + \lceil w/c \rceil -1,& \mbox{otherwise}\\
\end{array}
\right.
\end{equation}

Now define $I_x \subseteq [0,c)$ as follows:

\begin{equation}
I_x = 
\left\{
\begin{array}{ll}
(0,\, c-\bar d]				& \mbox{if } x=0\\
(c-\bar d -x,\, c -x]	& \mbox{if } 0 < x < c - \bar d\\
(c- x,\,  c - \bar d+ c - x] & \mbox{if } c- \bar d \le x < c
\end{array}
\right.
\end{equation}

Suppose that $r_{i'}$ and $r_{j'}$ are any two $r_{d'}$ representatives in $H$ and an 
$AddW(w)$ operation has just been performed.  

If $w \bmod c = 0$ then 
the relative ordering of $r_{i'}$ and $r_{j'}$ does not change so
$H$ remains unchanged and nothing further needs to be done.

Otherwise, let $x = \bar W_t \bmod c$ after the $AddW(w)$ operation. If $r_{i'}$ and $r_{j'}$ are either both in $I_x$ or both outside $I_x$ then they both get changed by the same amount and their relative ordering remains unchanged.  

In $H$ this means   the {\em max cost index} of internal node $v$
can only change if there are two leaf nodes in $I(v)$ such that one of the nodes is
in $I_x$ and one outside of $I_x$, i.e., $I_x \cap I(v) \not= \emptyset.$

We will now see that the number of such $v$ is $O(\log n)$ and their {\em max cost indices} can be updated in $O(\log n)$ total time.

 Set $i$ be the smallest leaf index in $H$ in $I_x$ and $j$ the largest leaf index in $H$ in $I_x$ and set $I = [i,j]$.  Note that $i,j$ can be found in $O(\log n)$ time using tree traversal.  By definition $ I_x \cap I(v) = I \cap I(v)$ so we will
restrict ourselves to examining $I \cap I(v)$.
 Let $h$ be the lowest common ancestor of $i',j'$ in $H$ and $r$ the root of $H$.

\medskip

By the construction of the tree 
it is easy to see that if $I(v) \cap I' \not = \emptyset$ then $v$ must be on (a) the path from $i$ to $h$, (b) the path from $j$ to $h$ or the path from $h$ to $r$.  Since
the tree has height $O(\log n)$ These paths together only contain $O(\log n)$ vertices.  

The algorithm is then to (a)  walk up 
the path from $i'$ to $h$ changing the max-cost values as appropriate (this can be done in $O(1)$ time per step since a node only needs to examine its (2 or 3) children, (b) similarly walk up the path from $j'$ to $h$ and then (c) from $h$ to $r.$ In total, this requires only $O(\log n)$ time.  After these walks are finished all internal nodes contain the correct max-cost values, so the algorithm is done.

\section{Proofs for Structure of Worst-Case Scenarios}
\subsection{Proof of Lemma \ref{lem:out-scenario}}
\label{app:out-scenario}
The lemma states if $s_B^*\in\mathcal{S}$ is transformed such that $w_i(s_B^*) = w_i^-$ if $x_i\notin P_d$, then it remains a \emph{worst-case scenario}. We prove this lemma by first assuming a \emph{worst-case scenario} $s^*\in\mathcal{S}$ and showing a transformation to $s_B^*$ with the required structure.

Consider a \emph{worst-case scenario} $s^*\in\mathcal{S}$ with dominant part $P_d$. Now construct a new scenario $s_B^*$ from $s^*$ by making the following changes: 
\begin{enumerate}
\item $w_i(s_B^*) = w_i(s^*)$ for $x_i\in P_d$, and 
\item $w_i(s_B^*)=w_i^-$ for $x_i\notin P_d$. 
\end{enumerate}

By making this change, the evacuation time for algorithm $A$'s choice of $\{\hat{P},\hat{Y}\}$ does not change since the scenario within the dominant part $P_d$ is unchanged. Thus,
\begin{equation} \label{eq:out-scenario-1}
\Theta^k(P,\{\hat{P},\hat{Y}\},s_B^*) = \Theta^k(P,\{\hat{P},\hat{Y}\},s^*)
\end{equation}

Since by making this change, we are only reducing the number of people in the
buildings, the optimal evacuation time can only decrease. Thus,
\begin{equation}
\label{eq:out-scenario-2} \Theta_{\mathrm{opt}}^k(P,s_B^*)\leq \Theta_{\mathrm{opt}}^k(P,s^*)
\end{equation}

By subtracting equations \ref{eq:out-scenario-1} and \ref{eq:out-scenario-2}, we can see that the regret under scenario
$s_B^*$ cannot be lesser than in $s^*$. This implies that $s_B^*$ is also a \emph{worst-case scenario} with dominant part $P_d\in \hat{P}$ such that $w_i(s_B^*) = w_i^-$ if $x_i\notin P_d$.\qed

\subsection{Proof of Lemma \ref{lem:sub-scenario}}
\label{app:sub-scenario}
Consider any arbitrary \emph{worst-case scenario} $s^*\in\mathcal{S}$ with dominant part $P_d\in \hat{P}$. We will show a construction from $s^*$ to $s_B^*\in\mathcal{S}$ such that the sub-scenario within $P_d$ lies in $\mathcal{S}_L^d\bigcup \mathcal{S}_R^d$. 

The sink chosen by algorithm $A$ in part $P_d$ is $y_d$. Let $x_{l_d}$(resp. $x_{r_d}$) be the leftmost(resp. rightmost) node in $P_d$. Without loss of generality, let us assume that $\Theta_L(P_d,y_d,s^*) > \Theta_R(P_d,y_d,s^*)$, i.e., the left evacuation time dominates in $P_d$ for sink $y_d$ under \emph{worst-case scenario} $s^*$. By the equation defined by Eq. \ref{eq:left-evac},
\begin{equation*} \Theta_{L}(P_d,y_d,s^*)=\max_{l_d\leq i\leq r_d}\left\{\left.
(x-y_d)\tau+\sum_{l_d\leq j\leq i}w_{j}(s^*)\right| y_d > x_i \right\}
\end{equation*}
let $m$ be the index of the node which maximizes the above term.

We will now show the transformation from $s^*$ to $s_B^*$. Consider a node $x_t$. If $x_t\in \{x_{m+1}, ..., x_{r_d}\}$, then by making $w_t(s_B^*) = w_y^-$,
\begin{align*}
\Theta_L(P_d,y_d,s_B^*) & = \Theta_L(P_d,y_d,s^*) && \\
\Theta_{\mathrm{opt}}^k(P,s_B^*) & \leq \Theta_{\mathrm{opt}}^k(P,s^*) &&
\end{align*}
In other words, the evacuation time is unchanged and the optimal $k$-sink evacuation time cannot increase. Thus, the regret in $s_B^*$ cannot be greater than the regret in $s^*$ after this transformation.

If $x_t \in \{x_{l_d}, ... ,x_m\}$, then by making $w_t(s_B^*) = w_t^+$ (See Fig. \ref{fig:left-dominant-scenario}),
\begin{align*}
\Theta_L(P_d,y_d,s_B^*) & = \Theta_L(P_d,y_d,s^*) + \left(w_t^+ - w_t(s^*)\right) && \\
\Theta_{\mathrm{opt}}^k(P,s_B^*) & \leq \Theta_{\mathrm{opt}}^k(P,s^*) + \left(w_t^+ - w_t(s^*)\right) &&
\end{align*}

\begin{figure}[tp]
\centering
\includegraphics[scale=0.4]{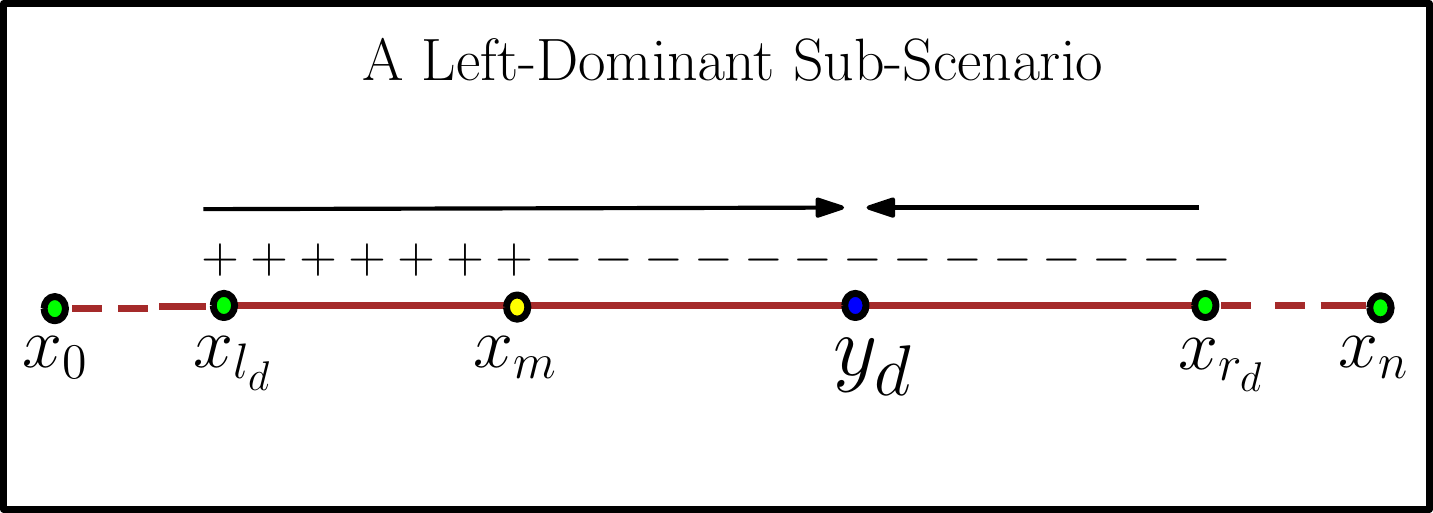}
\caption{\label{fig:left-dominant-scenario} A \emph{left-dominant} \emph{worst-case scenario} in path $P_d$}
\end{figure}

Thus, again by this transformation the regret in $s_B^*$ cannot be greater than the regret in $s^*$. By these two transformations, we have constructed a \emph{left-dominant} sub-scenario within $P_d$. A similar proof exists for constructing a \emph{right-dominant} sub-scenario when the right evacuation time dominates.

Therefore, the sub-scenario in $P_d$ lies in the set $\mathcal{S}_L^d\bigcup \mathcal{S}_R^d$. \qed

\subsection{Proof of Property \ref{prop:n-wcscenarios}}
\label{app:n-wcscenarios}
We are given the $k$-partition $\hat{P}=\{P_1,P_2,...,P_k\}$ by algorithm $A$. If we know the dominant part $P_d\in\hat{P}$, then Thm. \ref{thm:wcscenarios} established the structure of candidate \emph{worst-case scenarios}. 

Let $\ell_i$ be the number of nodes in part $P_i\in\hat{P}$. It is easy to see that $\sum_{1\leq i\leq k}(\ell_i) = n$. If $P_d$ were the dominant part, then there are a possible $O(\ell_d)$ \emph{worst-case scenarios} (from Thm. \ref{thm:wcscenarios}). Thus, the total number of candidate worst-case scenarios (assuming each $P_i$ to be the dominant part) is $O(\sum_{1\leq i\leq k}(\ell_i))$ which is $O(n)$. \qed

\subsection{Proof of Property \ref{prop:n2-wcscenarios}}
\label{app:n2-wcscenarios}
Irrespective of the choice of $\{\hat{P},\hat{Y}\}$ by algorithm $A$, the \emph{worst-case scenario} chosen by adversary $B$, $s_B^*\in\mathcal{S}^*$, is always of the form described in Thm. \ref{thm:wcscenarios}. Looked at from the perspective of path $P$, for some pair of integers $(t_1,t_2)$ such that $(0\leq t_1\leq t_2\leq n+1)$, $s_B^*$ can be broken down as:
\begin{align*}
s_B^* = 
\begin{cases}
w_{i}^{-} & ,0\leq i<t_{1}\\
w_{i}^{+} & ,t_{1}\leq i<t_{2}\\
w_{i}^{-} & ,t_{2}\leq i\leq n
\end{cases}
\end{align*}

There are $O(n^2)$ possible values for $(t_1,t_2)$. Therefore, there are $O(n^2)$ possible \emph{worst-case scenarios}, i.e., $\left| \mathcal{S}^* \right| = O(n^2)$.

\section{Proofs for the Characterization of Minimax-Regret}
\subsection{Proof of Lemma \ref{lem:rij-subpath-inc}}
\label{app:rij-subpath-inc}
Consider the part $P_i$ with left end(resp. right end) as vertex $x_{l_i}$(resp. $x_{r_i}$). Let $P_i^{\prime}$ be the part with one node appended to the right of part $P_i$, its left end is $x_{l_i}$ and its right end is $x_{r_i+1}$. From Eq. \ref{eq:mmr-it-sinks}, the minimax-regret for ADP(Assumed Dominant Part) $P_i$ and ADP $P_i^{\prime}$ can be written as:
\begin{align*}
R_{l_ir_i} & = \min_{l_i\leq t\leq r_i}\left\{R_{l_ir_i}(x_t)\right\} && \\
R_{l_i(r_i+1)} & = \min_{l_i\leq t\leq r_i+1}\left\{R_{l_i(r_i+1)}(x_t)\right\} &&
\end{align*}

We associate $R_{l_ir_i}$ with a sequence of terms $Z = \left(R_{l_ir_i}(x_{l_i}),R_{l_ir_i}(x_{l_i+1}),...,R_{l_ir_i}(x_{r_i})\right)$ where $R_{l_ir_i}$ is the minimum of the terms in the sequence $Z$. Similarly, we associate $R_{l_i(r_i+1)}$ with sequence $Z^{\prime} = \left(R_{l_i(r_i+1)}(x_{l_i}),R_{l_i(r_i+1)}(x_{l_i+1}),...,R_{l_i(r_i+1)}(x_{r_i+1})\right)$. Note that $Z^{\prime}$ has one more term than $Z$, which is $R_{l_i(r_i+1)}(x_{r_i+1})$.

\begin{claim}
For subpaths $P_i$ and $P_i^{\prime}$, $R_{l_ir_i}(x_t)\leq R_{l_i(r_i+1)}(x_t)$ for $(l_i\leq t \leq r_i)$ and $R_{l_ir_i}(x_{r_i})\leq R_{l_i(r_i+1)}(x_{r_i+1})$. In other words, we claim that every element in sequence $Z^{\prime}$ is greater than or equal to some element in $Z$.
\end{claim}
\begin{proof}
Consider some node $x_t (l_i\leq t\leq r_i)$. From Eq. \ref{eq:max-regret1}, we know that
\begin{equation*}
R_{l_ir_i}(x_t) = \max_{s\in\mathcal{S}^*}R_{l_ir_i}(s,x_t)
\end{equation*}
where $R_{l_ir_i}(s,x_t)$ is the regret with ADP $P_i$ for sink $x_t$ under scenario $s$. Let $s^*$ be the \emph{worst-case scenario} which maximizes the term $R_{l_ir_i}(s,x_t)$. Since $P_i$ is a subpath of $P_i^{\prime}$, the evacuation time for sink $x_t$ under scenario $s^*$ in $P_i$ is lesser than or equal to that in $P_i^{\prime}$, 
\begin{align*}
\Theta^1(P_i,x_t,s^*) & \leq \Theta^1(P_i^{\prime},x_t,s^*) && \\
\Theta^1(P_i,x_t,s^*) - \Theta^k_{\mathrm{opt}}(P,s^*) & \leq \Theta^1(P_i^{\prime},x_t,s^*) - \Theta^k_{\mathrm{opt}}(P,s^*) && \\
R_{l_ir_i}(x_t,s^*) & \leq R_{l_i(r_i+1)}(x_t,s^*) &&
\end{align*}
Since, $s^*$ is the maximizing term for $R_{l_ir_i}(x_t)$,
\begin{align}
R_{l_ir_i}(x_t) & \leq R_{l_i(r_i+1)}(x_t,s^*) && \nonumber\\
& \leq \max_{s\in\mathcal{S}^*}R_{l_i(r_i+1)}(x_t,s) \nonumber && \\
& = R_{l_i(r_i+1)}(x_t) && \label{eq:lem-rij-subpath-1}
\end{align}

Thus, for $x_t (l_i\leq t\leq r_i)$, $R_{l_ir_i}(x_t) \leq R_{l_i(r_i+1)}(x_t)$. Consider $s^*$ to be the \emph{worst-case scenario} which maximizes the regret $R_{l_ir_i}(s,x_{r_i})$. Under scenario $s^*$, for sink $x_{r_i}\in P_i$, the evacuation time in $P_i$ is equal to the left evacuation time as there are no people to evacuate to the right. Similarly, for sink $x_{r_i+1}\in P_i^{\prime}$, the evacuation time in $P_i^{\prime}$ equals the left evacuation time. But since $x_{r_i+1}$ lies to the right of $x_{r_i}$, the left evacuation time for $x_{r_i+1}$ will be greater than or equal to the left evacuation time for $x_{r_i}$, i.e.,
\begin{align*}
\Theta^1(P_i,x_{r_i},s^*) & \leq \Theta^1(P_i^{\prime},x_{r_i+1},s^*) && \\
\Theta^1(P_i,x_{r_i},s^*) - \Theta^k_{\mathrm{opt}}(P,s^*) & \leq \Theta^1(P_i^{\prime},x_{r_i+1},s^*) - \Theta^k_{\mathrm{opt}}(P,s^*) && \\
R_{l_ir_i}(x_{r_i},s^*) & \leq R_{l_i(r_i+1)}(x_{r_i+1},s^*) &&
\end{align*}
Since, $s^*$ is the maximizing term for $R_{l_ir_i}(x_{r_i})$,
\begin{align}
R_{l_ir_i}(x_{r_i}) & \leq R_{l_i(r_i+1)}(x_{r_i+1},s^*) && \nonumber \\
& \leq \max_{s\in\mathcal{S}^*}R_{l_i(r_i+1)}(x_{r_i+1},s) \nonumber && \\
& = R_{l_i(r_i+1)}(x_{r_i+1}) && \label{eq:lem-rij-subpath-2}
\end{align}

From Eqs. \ref{eq:lem-rij-subpath-1} and \ref{eq:lem-rij-subpath-2}, we can see that all the terms in sequence $Z^{\prime}$ is greater than or equal to some term in sequence $Z$.\qed
\end{proof}

\begin{corollary}
The minimum of the terms in sequence $Z$ ($=R_{l_ir_i}$) is lesser than or equal to the minimum of the terms in $Z^{\prime}$ ($=R_{l_i(r_i+1)}$) (Follows from the previous claim).
\end{corollary}

Thus, $R_{l_ir_i}\leq R_{l_i(r_i+1)}$ and by a symmetrical argument, $R_{l_ir_i}\leq R_{(l_i-1)r_i}$.\qed

\subsection{Proof of Lemma \ref{lem:rijs-subpath-inc}}
\label{app:rijs-subpath-inc}
Let $\mathcal{R}=\max_{1\leq i\leq q}\left\{R_{l_ir_i}\right\}$ and $\mathcal{R}^{\prime}=\max_{1\leq i\leq q}\left\{R_{l_i^{\prime} r_i^{\prime}} \right\}$. We have to prove that $\mathcal{R}\leq \mathcal{R}^{\prime}$.

Assume the contrary, i.e., $\mathcal{R} > \mathcal{R}^{\prime}$, i.e.,
\begin{equation*}
\max_{1\leq i\leq q}\left\{R_{l_ir_i}\right\} > \max_{1\leq i\leq q}\left\{R_{l_i^{\prime} r_i^{\prime}} \right\}
\end{equation*}

Now, let us consider the $q$-partition $\hat{P}_q=\{P_1,P_2,...,P_q\}$. Now, consider the $q$-partition of $\hat{P}_q^{\prime}$ but only upto $x_{r_q}$, i.e., restricting the $q^{th}$ part to end at $x_{r_q}$. Let us call this $\hat{P}_q^{\prime\prime}$. In the $q$-partition $\hat{P}_q^{\prime\prime}$, the $q^{th}$ part is smaller than the $q^{th}$ part in $\hat{P}_q^{\prime}$ (by 1 node). Therefore, by Lemma \ref{lem:rij-subpath-inc}, the minimax-regret for the $q^{th}$ part being the ADP in $\hat{P}_q^{\prime\prime}$ is lesser than or equal to the minimax-regret for the $q^{th}$ part as ADP in $\hat{P}_q^{\prime}$. All the other parts in $\hat{P}_q^{\prime\prime}$ have the same ADP minimax-regret as the parts in $\hat{P}_q^{\prime}$. Therefore, if $\mathcal{R}^{\prime\prime}$ denotes the minimax-regret for the $q$-partition $\hat{P}_q^{\prime\prime}$, then $\mathcal{R}^{\prime\prime}\leq \mathcal{R}^{\prime}$. This implies that $\mathcal{R}^{\prime\prime} < \mathcal{R}$, which contradicts our assumption that $\mathcal{R}$ was the minimax-regret considering the first $q$ parts as ADPs with right end of the $q^{th}$ part fixed at $x_{r_q}$. 

Therefore, $\mathcal{R}\leq \mathcal{R}^{\prime}$.

\section{Proofs for Minimax-Regret Binary-Search Based Algorithm}

\subsection{Calculation of $R_{l_ir_i}$ for all $P_i$ in $O(kn^2\log n)$ time}
\label{app:rliri-procedure}
From Property \ref{prop:n-wcscenarios}, we know that there are $O(n)$ candidate \emph{worst-case scenarios} given the $k$-partition $\hat{P}$. Let $\mathcal{S}_B^*$ be the set of these $O(n)$ scenarios. In order to find the minimax-regret given the $k$-partition $\hat{P}$, we can use the following simple procedure:
\begin{enumerate}
\item For every scenario $s_B^*\in \mathcal{S}_B^*$ ($O(n)$ scenarios):
  \begin{itemize}
    \item Find the optimal $k$-sink evacuation solution for $s_B^*$. ($O(n.k.\log n)$ time) and store in $\Theta_{\mathrm{opt}}^k(P,s_B^*)$.
  \end{itemize}
\item For every part $P_i (0\leq i \leq k)$:
  \begin{enumerate}
    \item For every candidate scenario $s_B^*\in\mathcal{S}_B^*$ ($O(n)$ scenarios by Property \ref{prop:n-wcscenarios}):
      \begin{itemize}
        \item Find $R_{l_ir_i}(s_B^*,x_j) \forall x_j\in P_i$. (in $O(n)$ time, see Appendix \ref{app:rliri-algo})
      \end{itemize}
  \end{enumerate}
\item For every possible sink $x_j (0\leq j\leq n)$ ($O(n)$ sinks):
  \begin{itemize}
    \item Find $R_{l_ir_i}(x_j)$ by finding $\max_{s_B^*\in\mathcal{S}_B^*} R_{l_ir_i}(s_B^*,x_j)$. ($O(n)$ time).
  \end{itemize}
\item For every part $P_i (0\leq i \leq k)$:
  \begin{itemize}
    \item Find $R_{l_ir_i}$ by finding $\min_{l_i \leq j\leq r_i}R_{l_ir_i}(x_j)$. ($O(n)$ time for all parts combined)
  \end{itemize}
\item The minimax-regret is the maximum of all $R_{l_ir_i}$ values. ($O(n)$ time)
\end{enumerate}

As we can see, this simple procedure to find the minimax-regret $R_{l_ir_i}$ for all $P_i$ given the $k$-partition $\hat{P}$ takes $O(kn^2\log n)$ time.

\subsection{Proof of Property \ref{prop:unimodality-partitions}}
\label{app:unimodality-partitions}
If the parts $\{P_{i+1},P_{i+2},...,P_k\}$ are fixed, then consider the left end of the $i^{th}$ part to be at $x_{l_i}$. If the minimax-regret for the entire path is $\mathcal{R}$, we can write it as:
\begin{align}
\mathcal{R} & = \max_{1\leq j\leq k}\left\{R_{l_jr_j}\right\} && \nonumber \\
& = \max\left(\max_{1\leq j\leq i-1}R_{l_jr_j},\max_{i\leq j\leq k}R_{l_jr_j}\right) \label{eq:unimodal-r} &&
\end{align}

Now if we move the left end of the $i^{th}$ part to $x_{l_i+1}$, then by Lemma \ref{lem:rij-subpath-inc}, we know that $R_{(l_i+1)r_i}\leq R_{l_ir_i}$. The values of $R_{l_jr_j}$ for $j=i+1,...,k$ do not change as those parts do not change. This means the term $\max_{i\leq j\leq k}R_{l_jr_j}$ in Eq. \ref{eq:unimodal-r} cannot increase. Also, by Lemma \ref{lem:rijs-subpath-inc}, we know that the minimax-regret considering only the first $i-1$ parts as ADPs cannot decrease. This means that in Eq. \ref{eq:unimodal-r}, the term $\max_{1\leq j\leq i-1}R_{l_jr_j}$ cannot decrease. Therefore, $\mathcal{R}$ is unimodal as a function of the left end of the $i^{th}$ part if the parts $\{P_{i+1},...,P_k\}$ are fixed.\qed 

\subsection{Algorithm to find $R_{l_ir_i}(s_B^*,x_j) \forall x_j\in P_i$ in $O(n)$ time}
\label{app:rliri-algo}
Given a \emph{worst-case scenario} $s_B^*\in\mathcal{S}_B^*$ (or any arbitrary scenario in fact) and an Assumed Dominant Part(ADP) $P_i\in\hat{P}$ with left end(resp. right end) as $x_{l_i}$(resp. $x_{r_i}$), we first give a simple procedure to calculate $\Theta^1(P_i,x_j,s_B^*)\forall j (l_i\leq j\leq r_i)$, the $1$-sink evacuation time in $P_i$ for all possible sinks $x_j\in P_i$ under scenario $s_B^*$.

We first prove the following lemma which will help us in proposing an efficient procedure:
\begin{lemma} \label{lem:rliri-2choices}
Consider a subpath(part) $P_i\in\hat{P}$ and a sink $x_j\in P_i$. Let us consider the left evacuation time which is the maximum of the evacuation functions of the vertices to the left of $x_j$ (from Eq. \ref{eq:left-evac}). Let $m$(resp. $x_m$) be the index(resp. vertex) which maximizes this function. If we move the sink from $x_j$ to $x_{j+1}$, then the maximizing index(resp. vertex) for the left evacuation to $x_{j+1}$ would be either $m$(resp. $x_m$) or $j$(resp. $x_j$). The same statement holds symmetrically for the right evacuation time.
\end{lemma}
\begin{proof}
From Eq. \ref{eq:left-evac}, the left evacuation time for sink $x_j\in P_i$ is:
\begin{equation*}
\Theta_{L}(P_i,x_j,s_B^*)=\max_{l_i\leq t\leq r_i}\left\{\left. (x_j-x_t)\tau+\sum_{l_i\leq z\leq t}w_{z}(s_B^*)\right| x_j > x_t \right\}
\end{equation*}

and the maximizing term is $t=m$. Now, when the sink is moved to $x_{j+1}$, notice that the only change in the evacuation function of a node is the addition of the length $\left|x_{j+1}-x_j\right|$. Therefore, the evacuation fuction among the nodes $\{x_{l_i},...,x_{j-1}\}$ is still maximum for $x_m$. Also, the evacuation function of one new node $x_j$ (the previous sink) is added as a candidate for the maximum term. Therefore, the new maximum can be calculate by one comparison of the evacuation function between the old maximum (for $x_m$) and newly added term (for $x_j$). \qed
\end{proof}

The following procedure will find $\Theta^1(P_i,x_j,s_B^*)\forall j (l_i\leq j\leq r_i)$ given scenario $s_B^*$:
\begin{enumerate}
  \item For $x_j\in P_i$ (Iterating from $j=l_i$ to $r_i$):
    \begin{itemize}
      \item Find left evacuation time for $x_j$ in $O(1)$ time (Lemma \ref{lem:rliri-2choices})
    \end{itemize}
  \item For $x_j\in P_i$ (Iterating back from $j=r_i$ to $l_i$):
    \begin{itemize}
      \item Find right evacuation time for $x_j$ in $O(1)$ time (Lemma \ref{lem:rliri-2choices})
    \end{itemize}
  \item For $x_j\in P_i$:
    \begin{itemize}
      \item Use the max of left and right evacuation times to find the evacuation time for $x_j$ under scenario $s_B^*$.
    \end{itemize}
\end{enumerate}
From this procedure, we can find out $\Theta^1(P_i,x_j,s_B^*)$, the evacuation time in $P_i$ under sink $x_j\forall j (l_i\leq j\leq r_i)$ in $O(n)$ time. Since we already know $\Theta^k_{\mathrm{opt}}(P,s_B^*)$ (from the previous step of the procedure in Sec. \ref{sec:algo-bs}), the optimal $k$-sink evacuation time for $s_B^*$, $R_{l_ir_i}(s_B^*,x_j)$ ($=\Theta^1(P_i,x_j,s_B^*)-\Theta^k_{\mathrm{opt}}(P,s_B^*)$) can be determined in $O(n)$ time.

\section{An $O(kn)$ DP Implementation for the Minimax-Regret $k$-Sink Location Problem}
\label{app:mmr-k-sink-dp}

\subsection{DP Algorithm for Minimax-Regret $k$-Sink Location}
A naive DP table filling procedure of the recurrence in Eq. \ref{eq:rec} will take $O(kn^2)$ time. Also, the minimax-regret $k$-partition and sinks can be reconstructed by storing the optimizing values during the DP without adding any extra time complexity. We observe two crucial properties (very similar to the Properties \ref{prop:dp-1} and \ref{prop:dp-2}) which helps bring down the running time.

\begin{property} \label{prop:mmr-dp-1}
In the recurrence given by Eq. \ref{eq:rec}, keeping $q$(number of parts) fixed, if we increment $i$, then the minimizing $j$ value cannot decrease.
\end{property}
\begin{proof}
See Appendix \ref{app:mmr-dp-1}
\end{proof}

\begin{property} \label{prop:mmr-dp-2}
In the recurrence given by Eq. \ref{eq:rec}, keeping $i$ and $q$ fixed, $\max\left(M(q-1,j-1),R_{ji}\right)$ is unimodal with a unique minimumm value as a function of $j$.
\end{property}
\begin{proof}
See Appendix \ref{app:mmr-dp-2}
\end{proof}

Therefore, because of the two properties stated above, we observe that Remark \ref{rem:faster-dp} also holds true for this recurrence thereby allowing us to update $j$ in an amortized $O(1)$ time. The new recurrence is:
\begin{equation}
M(q,i) = \min_{j^{\prime}\leq j\leq i}\max\left(M(q-1,j-1),R_{ji}\right)
\end{equation}
where $j^{\prime}$ is the optimum value of $j$ for $M(q,i-1)$.

The DP procedure is illustrated below:
\begin{enumerate}
  \item Fill in $M(1,i) = R_{0i} (0\leq i\leq n)$. ($O(n)$ time)
  \item For $q\leftarrow 2$ to $k$:
    \begin{itemize}
      \item $j\leftarrow 0$
    \end{itemize}
    \begin{enumerate}
      \item For $i\leftarrow 0$ to $n$:
        \begin{itemize}
          \item Do $j\leftarrow j+1$ till minimum of $\max(M(q-1,j-1),R_{ji})$ is found which is equal to $M(q,i)$.
          \item Store $M(q,i)$ in the DP table.
        \end{itemize}
    \end{enumerate}
\end{enumerate}

For a given value of $q$(number of parts), $i$ and $j$ are incremented atmost $n$ times each and each increment is handled in $O(1)$ time. Therefore, the running time is $O(kn)$.

\subsection{Proof of Property \ref{prop:mmr-dp-1}}
\label{app:mmr-dp-1}
The proof for this property is very similar to the proof for Property \ref{prop:dp-1} because of the similar nature the optimal evacuation time recurrence and the minimax-regret recurrence.

Equation \ref{eq:rec} is given by:
\begin{equation}
M(q,i) = \min_{0\leq j\leq i}\max\left(M(q-1,j-1),R_{ji}\right)
\end{equation}
where $R_{ji}$ is the minimax-regret with ADP(Assumed Dominant Part) as the part with left end $x_j$ and right end $x_i$.

Assume $j^{\prime}$ is the  minimizing $j$ value for $M(q,i-1)$ and $j^{\prime\prime}$ is the minimizing $j$ value for $M(q,i)$. We need to prove that $j^{\prime\prime}\geq j^{\prime}$.

Assume the contrary, i.e., $j^{\prime\prime}<j^{\prime}$. We now have two cases to deal with:\\\\
\textbf{Case 1}: $R_{j^{\prime}(i-1)} \geq M(q-1,j^{\prime}-1)$

In this case, the minimax-regret calculated in the subpath from $x_0$ to $x_{i-1}$ for $q$ parts will be $M(q,i-1) = R_{j^{\prime}(i-1)}$. Now,
\begin{equation*}
M(q,i) = \max\left(M(q-1,j^{\prime\prime}),R_{j^{\prime\prime}i}\right)
\end{equation*}

\begin{remark} \label{rem:path-subpath-mmr}
Consider the minimax-regret considering only the first $q$ parts as ADPs. This minimax-regret calculated when these first $q$ parts are restricted to a path is greater than or equal to the minimax-regret when the first $q$ parts are restricted to any of its subpaths.
\end{remark}
\begin{proof}
This follows from Lemma \ref{lem:rijs-subpath-inc}, where it is proved when a path is extended to the right by a node. By symmetry, it is also true when the path is extended to the left.\qed
\end{proof}

Since $j^{\prime\prime} < j^{\prime}$, 
\begin{align}
M(q-1,j^{\prime\prime}-1) & \leq M(q-1,j^{\prime}-1) \nonumber && \text{(from Remark \ref{rem:path-subpath-mmr})}\\
 & \leq R_{j^{\prime}(i-1)} \nonumber && \text{(from case assumption)} \\
 & \leq R_{j^{\prime}i} \nonumber && \text{(from Remark \ref{rem:path-subpath-mmr})}\\
& \leq R_{j^{\prime\prime}i} && \text{(from Remark \ref{rem:path-subpath-mmr})} \label{eq:prop-mmr-1}
\end{align}

Therefore from Eq. \ref{eq:prop-mmr-1}, $M(q,i)=R_{j^{\prime\prime}i}$. Since $j^{\prime\prime}$ is the largest minimizing $j$ value for $M(q,i)$ and $j^{\prime} > j^{\prime\prime}$,
\begin{align}
R_{j^{\prime\prime}i} & < \max\left(M(q-1,j^{\prime}-1),R_{j^{\prime}i}\right) && \label{eq:prop-mmr-2} \\
& = R_{j^{\prime}i}
\end{align}
But from Lemma \ref{lem:rij-subpath-inc}, we know that $R_{j^{\prime\prime}i}\geq R_{j^{\prime}i}$. Eq. \ref{eq:prop-mmr-2} cannot be true. We arrive at a contradiction. Therefore, our assumption that $j^{\prime\prime} < j^{\prime}$ is false for this case. So,  $j^{\prime\prime}\geq j^{\prime}$.\\\\
\textbf{Case 2}: $R_{j^{\prime}(i-1)} < M(q-1,j^{\prime}-1)$

In this case, $M(q,i-1) = M(q-1,j^{\prime}-1)$. The minimizing $j$ value for $M(q,i)$ is $j^{\prime\prime}$,
\begin{equation*}
M(q,i) = \max\left(M(q-1,j^{\prime\prime}-1),R_{j^{\prime\prime}i}\right)
\end{equation*}

Since $j^{\prime\prime} < j^{\prime}$,
\begin{equation} \label{eq:prop-mmr-3}
\max\left(M(q-1,j^{\prime\prime}-1),R_{j^{\prime\prime}i}\right) < \max\left(M(q-1,j^{\prime}-1),R_{j^{\prime}i}\right)
\end{equation}

By the case assumption and the fact that $j^{\prime}$ is the minimizing $j$ value for $M(q,i-1)$,
\begin{align}
R_{j^{\prime}(i-1)} & < M(q-1,j^{\prime}-1) \nonumber && \text{(from case assumption)}\\
& \leq \max(M(q-1,j^{\prime\prime}-1),R_{j^{\prime\prime}(i-1)}) \nonumber && \\
& \leq \max(M(q-1,j^{\prime\prime}-1),R_{j^{\prime\prime}i}) \label{eq:prop-mmr-4} && \text{(from Remark \ref{rem:path-subpath-mmr})}
\end{align}
Therefore,
\begin{align}
M(q-1,j^{\prime}-1) & \leq \max(M(q-1,j^{\prime\prime}-1),R_{j^{\prime\prime}i}) \label{eq:prop-mmr-5} && \text{(from Eq. \ref{eq:prop-mmr-4})}\\
R_{j^{\prime}i} & \leq R_{j^{\prime\prime}i} && \label{eq:prop-mmr-6} \text{(from Remark \ref{rem:path-subpath-mmr})}
\end{align}
It is easy to see that Eqs. \ref{eq:prop-mmr-3}, \ref{eq:prop-mmr-5} and \ref{eq:prop-mmr-6} contradict each other. Therefore, our assumption that $j^{\prime\prime} < j^{\prime}$ is false for this case. So, $j^{\prime\prime} \geq j^{\prime}$. \qed

\subsection{Proof of Property \ref{prop:mmr-dp-2}}
\label{app:mmr-dp-2}
The minimax-regret is given by the recurrence in Eq. \ref{eq:rec}:
\begin{equation*}
M(q,i) = \min_{0\leq j\leq i}\left\{\max\left(M(q-1,j-1),R_{ji}\right)\right\}
\end{equation*}

Keeping $q$ and $i$ fixed, if we increment $j$, then by Lemma \ref{lem:rij-subpath-inc}, $R_{ji}$ cannot increase and by Lemma \ref{lem:rijs-subpath-inc}, $M(q-1,j-1)$ cannot decrease. This implies that $\max\left(M(q-1,j-1),R_{ji}\right)$ is unimodal as a function of $j$.\qed

\section{Reduction of $R_{ji}$ Precomputation Time from $O(n^5)$ to $O(n^3)$}
\label{app:rji-precomputation}
In Sect. \ref{sec:rji-naive}, we had given an $O(n^5)$ procedure for the precomputation of $R_{ji}$'s. Here, we will reduce the precomputation time to $O(n^3)$.

\subsection{A Lookup Table for Step 4 - $O(n^4)$ \label{sec:precomp-amort-step2}}
In Step 4 of the Naive approach (in Sect. \ref{sec:rji-naive}), we take $O(n)$ time to compute the regret because of the calculation of the evacuation time on a subpath given a sink and a scenario. Instead, we can construct an $O(1)$ time lookup data structure in $O(n^3)$ time. A (subpath,sink,scenario) query to the structure will yield the $1$-sink evacuation time in the subpath for the sink under the scenario in $O(1)$ time (For details see Appendix \ref{app:1-sink-subpath-precomp}.) Now, finding the regret for a (subpath,sink,scenario) triplet can be done in $O(1)$ time. Thus, the time for precomputation of $R_{ji}$ comes down to $O(n^4)$.

\subsection{Amortizaton on Step 2 - $O(n^3)$}
In Step 2, for each subpath, we are checking $O(n)$ possible minimax-regret sink locations. Instead, if the following two lemmas were true, we would only need to check an amortized $O(1)$ sink locations.

\begin{lemma}
\label{lem:rji-right}
Consider an ADP(Assumed Dominant Part) with leftmost(resp. rightmost) node as $x_l$(resp. $x_r$). Let $x_t(l\leq t\leq r)$ be the minimax-regret sink which minimizes the \emph{max-regret}. In the subpath from $x_l$ to $x_{r+1}$ (extending the right end of the part), there exists a minimax-regret sink at some $x_i$, where $i\geq t$.
\end{lemma}
\begin{proof} See Appendix \ref{app:rji-right}.
\end{proof}

\begin{lemma}
\label{lem:rji-unimodal}
Consider an ADP with leftmost(resp. rightmost) node as $x_l$(resp. $x_r$). The \emph{max-regret} for a sink $x_i (l\leq i\leq r)$ is $R_{lr}(x_i)$. $R_{lr}(x_i)$ is unimodal with a unique minimum value as a function of the sink $x_i (l\leq i\leq r)$.
\end{lemma}
\begin{proof} See Appendix \ref{app:rji-unimodal}.
\end{proof}

By using Lemmas \ref{lem:rji-right} and \ref{lem:rji-unimodal}, we can see that by apppending a node to the right of a subpath, the minimax-regret sink cannot move to the left and it is possible to locate the minimax-regret sink in the subpath by scanning linearly to right of the previous minimax-regret sink. Thus, the precomputation of $R_{ji}$'s can be rewritten to be completed in $O(n^3)$ time as follows:
\begin{enumerate}
\item Consider some $x_l$ as the left end of the subpath, move the right-end $x_r$ away from $x_l$. There are $O(n)$ possible right-ends to a left end $x_l$.
\item By Lemmas  \ref{lem:rji-right} and \ref{lem:rji-unimodal}, the minimax-regret sink $x_t$ cannot move to the left as the right-end $x_r$ increases. Therefore over all possible right-ends $x_r$, $t$(index of current minimax-regret sink) is incremented only $O(n)$ times, i.e., only an amortized $O(1)$ sink locations need to be checked for a right-end $x_r$.
\item Each candidate minimax-regret sink has $O(n)$ possible \emph{worst-case scenarios} $s_B^*$. The regret can be looked up in $O(1)$ time because of the precomputation in Sect. \ref{sec:precomp-amort-step2}. The \emph{max-regret} for the sink over all possible \emph{worst-case scenarios} $s_B^*$ can be calculated in $O(n)$ time.
\item There are $O(n)$ possible left ends $x_l$ to consider.
\item Therefore, the total running time in $O(n^3)$.
\end{enumerate}

\section{Proofs for the precomputation of $R_{ji}$'s}

\subsection{Construction of Lookup Table for the $1$-sink evacuation time for any (subpath,sink,scenario) triplet in $O(n^3)$ time}
\label{app:1-sink-subpath-precomp}
We are going to show a method by which we can construct a $O(1)$ time lookup structure for the evacuation time of any subpath with left end $x_l$ and right end $x_r$, for a sink  $x_t (l\leq t\leq r)$ within the subpath under all candidate \emph{worst-case scenarios} in $s_B^*\in \mathcal{S}^*$.

Now, according to the proof of Lemma \ref{lem:sub-scenario} given in Appendix \ref{app:sub-scenario}, the \emph{worst-case scenario} for a sink in a subpath is either \emph{left-dominant} or \emph{right-dominant} and the weights of the vertices transform from $w_i^-$ to $w_i^+$ at a node $x_m$ which maximizes the \emph{evacuation function} for that sink (Refer Fig. \ref{fig:left-dominant-scenario}). Therefore, a \emph{left-dominant} or \emph{right-dominant} scenario in a subpath from $x_l$ to $x_r$ with sink $x_t$ can be represented by a vertex $x_m$ which is the transition point of the weights from $w_i^-$ to $w_i^+$.

We are going to construct two tables:
\begin{itemize}
\item $\mathcal{L}(x_l,x_m,x_t)$ - Stores the evacuation time in subpath from $x_l$ to $x_t$ under a \emph{left-dominant} scenario with transition vertex $x_m (l\leq m < t)$.
\item $\mathcal{R}(x_t,x_m,x_r)$ - Stores the evacuation time in subpath from $x_t$ to $x_r$ under a \emph{right-dominant} scenario with transition vertex $x_m (t < m\leq r)$.
\end{itemize}

Conceptually, tables $\mathcal{L}$ and $\mathcal{R}$ store the left and right evacuation time for a sink, subpath and scenario. The evacuation time in a subpath from $x_l$ to $x_r$ for a sink $x_t (l\leq t\leq r)$ under a candidate \emph{worst-case scenario} which has transition vertex $x_m$ is:
\begin{itemize}
\item $\max\left(\mathcal{L}(x_l,x_m,x_t),\mathcal{R}(x_t,x_r,x_r)\right)$, if $x_m < x_t$ (a \emph{left-dominant} scenario), or
\item $\max\left(\mathcal{L}(x_l,x_l,x_t),\mathcal{R}(x_t,x_m,x_r)\right)$, if $x_m > x_t$ (a \emph{right-dominant} scenario).
\end{itemize}
Therefore, any (subpath,sink,scenario) query for the evacuation time in a subpath for a sink under a scenario can be looked up in $O(1)$ time if we have the tables $\mathcal{L}$ and $\mathcal{R}$.

We now show how to calculate the table $\mathcal{R}(x_t,x_m,x_r)$. Table $\mathcal{L}(x_l,x_m,x_t)$'s construction is symmetric. For a given sink $x_t$ and the right end $x_r$, $x_m (t <  m\leq r)$ has a possible $O(n)$ locations. The following procedure fills up the $\mathcal{R}$ table:

\begin{enumerate}
  \item For $t\leftarrow 0$ to $n-1$: ($O(n)$)
    \begin{enumerate}
      \item For $r\leftarrow t+1$ to $n$: ($O(n)$)
        \begin{itemize}
          \item Calculate $\mathcal{R}(x_t,x_{t+1},x_r)$ normally in $O(n)$ time.
        \end{itemize}
        \begin{enumerate}
          \item For $m\leftarrow t+2$ to $r$: ($O(n)$)
            \begin{itemize}
              \item $\mathcal{R}(x_t,x_m,x_r) = \mathcal{R}(x_t,x_{m-1},x_r) + \left|x_{m}-x_{m-1}\right| - w_{m-1}$
            \end{itemize}
        \end{enumerate}
    \end{enumerate}
\end{enumerate}

A similar procedure can be used to fill up the table $\mathcal{L}(x_l,x_m,x_t)$. The running time for construction of $\mathcal{L}$ and $\mathcal{R}$ is $O(n^3)$.

\subsection{Proof for Lemma \ref{lem:rji-right}}
\label{app:rji-right}
Let us first define a few notations we will be using in this proof:
\begin{itemize}
\item $P^O$ is the subpath with leftmost(resp. rightmost) node $x_l$(resp. $x_r$).
\item $P^N$ is the subpath with leftmost(resp. rightmost) node $x_l$(resp. $x_{r+1}$).
\item $x_t$(resp. $x_i$) is the the minimax-regret sink in subpath $P^O$(resp. $P^N$).
\item For node $x_j$ in subpath $P^O$(resp. $P^N$), $R^O_j$(resp. $R^N_j$) is the minimax-regret if $x_j$ is chosen as the sink.
\item $\Theta_L^O(s,x_j)$ (resp. $\Theta_L^N(s,x_j)$) is the left evacuation time in subpath $P^O$(resp. $P^N$) if $x_j$ is chosen as the sink under scenario $s$. Similarly, $\Theta_R^O(s,x_j)$ and $\Theta_R^N(s,x_j)$ is defined.
\item $\Theta^O(s,x_j)$ (resp. $\Theta^N(s,x_j)$) is the overall evacuation time in subpath $P^O$(resp. $P^N$) if $x_j$ is chosen as the sink under scenario $s$.
\end{itemize}

Let us assume the contrary, i.e., there exists no minimax-regret sink $x_i\in P^N$ such that $i\geq t$, i.e., $i < t$. Therefore,

\begin{equation}
\label{eq:l3:2}
R_i^N < R_t^N
\end{equation}
since the minimax-regret sink in $P^N$ is not located at $i=t$.

Since $x_t$ is the minimax-regret sink in $P^O$,
\begin{equation}
\label{eq:l3:1}
R_t^O\leq R_i^O
\end{equation}

Now, since $P^N$ is $P^O$ appended by a node, the following equation holds:
\begin{equation}
\label{eq:l3:3}
R_i^O\leq R_i^N
\end{equation}

Equations \ref{eq:l3:1},\ref{eq:l3:2} and \ref{eq:l3:3} gives us,
\begin{equation}
\label{eq:l3:4}
R_t^O < R_t^N
\end{equation}

Let $s_t^N$ be the \emph{worst-case scenario} for $x_t$ in subpath $P^N$.
\begin{claim}
\label{cl:l3:1}
$\Theta_R^N(s_t^N,x_t)\geq \Theta_L^N(s_t^N,x_t)$,i.e., in subpath $P^N$, for sink $x_t$, the right evacuation time is greater than or equal to the left evacuation time under its \emph{worst-case scenario} $s_t^N$.
\end{claim}
\begin{proof}
Assume the contrary, i.e., $\Theta_R^N(s_t^N,x_t) < \Theta_L^N(s_t^N,x_t)$. Now,
\begin{eqnarray*}
R_t^N & = & \Theta_L^N(s_t^N,x_t) - \Theta_{opt}(s_t^N)\\
& = & \Theta_L^O(s_t^N,x_t) - \Theta_{opt}(s_t^N)\\
& \leq & R_t^O
\end{eqnarray*}
This contradicts Equation \ref{eq:l3:4}. Our assumption is false and therefore, $\Theta_R^N(s_t^N,x_t)\geq \Theta_L^N(s_t^N,x_t)$.\qed
\end{proof}

Since $i < t$, by Claim \ref{cl:l3:1},
\begin{eqnarray*}
\Theta^N(s_t^N,x_t) & \leq & \Theta^N(s_t^N,x_i)\\
\implies R_t^N & \leq & R_i^N
\end{eqnarray*}
This contradicts Equation \ref{eq:l3:2}. Therefore our assumption that there exists no minimax-regret sink $x_i$ ($i\geq k$) is false. Thus there exists a minimax-regret sink $x_i$ in $P^N$ such that $i\geq t$.\qed

\subsection{Proof for Lemma \ref{lem:rji-unimodal}}
\label{app:rji-unimodal}
Consider $P_{lr}$ to be the subpath from node $x_l$ to $x_r$. For any given scenario $s\in\mathcal{S}$, the evacuation time in subpath $P_{lr}$ will be unimodal with unique minimum value as a function of the sink $x_i$, i.e.,$\Theta^1(P_{lr},x_i,s)$ is unimodal.

By definition (in Eq. \ref{eq:regret1}), $$R_{lr}(s,x_i) = \Theta^1(P_{lr},x_i,s)-\Theta_{\mathrm{opt}}^k(P,s).$$
Therefore, $R_{lr}(s,x_i)$ is also unimodal with a unique minimum value as a function of $x_i$. Also the \emph{max-regret} can be defined as (by Eq. \ref{eq:max-regret1}),
$$R_{lr}(x_i) = \max_{s\in\mathcal{S}^*}R_{lr}(s,x_i)$$

$R_{lr}(x_i)$ is the maximum of $O(n)$ unimodal functions with unique minimum values. Therefore, $R_{lr}(x_i)$ is also unimodal function with a unique minimum value as a function of $x_i (l\leq i\leq r)$.\qed

\end{appendix}

\end{document}